\documentclass[a4paper,UKenglish]{lipics-v2016}

\usepackage{thmtools}
\usepackage{thm-restate}
\usepackage{stmaryrd}
\usepackage{xspace}
\usepackage{etoolbox}
\usepackage{multidef}
\usepackage{Abreviations}
\usepackage{GameMacro}
\usepackage{Couleur}
\usepackage{SymbolCircleBoxETC}
\usepackage{wrapfig}
\usepackage{colonequals}
\usepackage{mathtools}
\PassOptionsToPackage{usenames,dvipsnames,svgnames,table}{xcolor}
\usepackage{nico}
\usepackage{tikz}
\usetikzlibrary{arrows,calc,shapes,snakes}
\usetikzlibrary{decorations.pathreplacing}

\title{Dependences in Strategy Logic}

\author[1]{Patrick Gardy}
\author[1]{Patricia Bouyer}
\author[1,2]{Nicolas Markey}
\affil[1]{LSV, CNRS \& ENS Paris-Saclay, Universit\'e Paris-Saclay,
  France}
\affil[2]{IRISA, CNRS \& INRIA \& Universit\'e Rennes 1, France}
\authorrunning{P. Gardy, P. Bouyer, N.Markey} 

\Copyright{Patrick Gardy, Patricia Bouyer, Nicolas Markey}


\DOIPrefix{}

\EventEditors{John Q. Open and Joan R. Acces}
\EventNoEds{2}
\EventLongTitle{42nd Conference on Very Important Topics (CVIT 2016)}
\EventShortTitle{CVIT 2016}
\EventAcronym{CVIT}
\EventYear{2016}
\EventDate{December 24--27, 2016}
\EventLocation{Little Whinging, United Kingdom}
\EventLogo{}
\SeriesVolume{42}
\ArticleNo{23}

\begin{document}
\hfuzz=3pt\relax  
  \maketitle
  
    
  \begin{abstract}
    Strategy Logic~(\SL) is a very expressive logic for specifying and
verifying properties of multi-agent systems: in~\SL, one can quantify
over strategies, assign them to agents, and express properties of the
resulting plays. Such a powerful framework has two drawbacks: first,
model checking \SL has non-elementary complexity; second, the exact
semantics of \SL is rather intricate, and may not correspond to what
is expected. In~this paper, we~focus on \emph{strategy dependences}
in~\SL, by tracking how existentially-quantified strategies in a
formula may (or~may~not) depend on other strategies selected in the
formula. We~study different kinds of dependences, refining the
approach of~[Mogavero \emph{et~al.}, Reasoning about strategies:
On~the model-checking problem,~2014], and prove that they give rise to
different satisfaction relations. In~the setting where strategies may
only depend on what they have observed, we~identify a large fragment
of~\SL for which we prove model checking can be performed
in \EXPTIME[2].

  \end{abstract}

  \section{Introduction}
  \label{se::Introduction}
  \subparagraph{Temporal logics.}
Since Pnueli's seminal paper~\cite{PnueliIntoLTL} in~1977, temporal
logics have been widely used in theoretical computer science,
especially by the formal-verification community. Temporal logics
provide powerful languages for expressing properties of reactive
systems, and enjoy efficient algorithms for satisfiability and model
checking~\cite{clarke1999model}.
%
Since the early 2000s, new temporal logics have appeared to
address \emph{open} and \emph{multi-agent systems}. While classical
temporal logics (e.g.~\CTL~\cite{CE81,QS82a}
and~\LTL~\cite{PnueliIntoLTL}) could only deal with one or all the
behaviours of the whole system, \ATL~\cite{IntroductionATLetStar}
expresses properties of (executions generated~by) behaviours of
individual components of the system.
\ATL~has been extensively studied since then, both
about its expressiveness and about its verification
algorithms~\cite{IntroductionATLetStar,GvD06,LMO08}. 

\subparagraph{Strategic interactions in \ATL.}  Strategies in~\ATL are
handled in a very limited way, and there are no real \emph{strategic
interactions} in that logic (which, in~return, enjoys a
polynomial-time model-checking algorithm).  Over the last 10 years,
various extensions have been defined and studied in order to allow for
more interactions~\cite{IATL2007,CHATTERJEE2010677,
ATLscWithBoundedMemory2009,Mogavero:2014:RSM:2656934.2631917,BSIL2011}.
\emph{Strategy Logic}~(\SL for
short)~\cite{CHATTERJEE2010677,Mogavero:2014:RSM:2656934.2631917} is
such a powerful approach, in which strategies are first-class
objects; formulas can quantify (universally and
existentially) over strategies, store those strategies in variables,
assign them to players, and express properties of the resulting
plays. As~a simple example, the existence of a winning strategy for
Player~$A$ (with objective~$\phi_A$) against any strategy of
Player~$B$ would be written as \( \exists \sigma_A.\ \forall\sigma_B.\
\textsf{assign}(A\mapsto \sigma_A; B\mapsto \sigma_B).\ \phi_A \).
This makes the logic both expressive and easy to use (at~first sight),
at the expense of a very high complexity: \SL~model checking has
non-elementary complexity, and satisfiability is
undecidable~\cite{Mogavero:2014:RSM:2656934.2631917,LM15}.

\subparagraph{Strategy dependences in~\SL.}
It~has been noticed in recent works that the nice expressiveness
of~\SL comes with unexpected phenomena. One recently-identified
phenomenon~\cite{bouyer2016semantics} is induced by the separation of
strategy quantification and strategy assignment: are the events
between strategy quantifications and strategy assignments part of the
\emph{memory} of the strategy? While both options may make sense
depending on the applications, only one of them makes model checking
decidable~\cite{bouyer2016semantics}.

\looseness=-1
A~second phenomenon---which is the main focus of the present
paper---concerns \emph{strategy
dependences}~\cite{Mogavero:2014:RSM:2656934.2631917}: in a formula
such as $\forall\sigma_A.\ \exists \sigma_B.\ \xi$, the
existentially-quantified strategy~$\sigma_B$ may depend on \emph{the
whole} strategy~$\sigma_A$; in~other terms, the action returned by
strategy~$\sigma_B$ after some finite history~$\rho$ may depend on
what strategy~$\sigma_A$ would play on any other
history~$\rho'$. Again, this may be desirable in some contexts, but
it~may also make sense to require that strategy~$\sigma_B$ after
history~$\rho$ can be
\emph{computed} based solely on
what has been observed along~$\rho$. This~approach was initiated
in~\cite{Mogavero:2014:RSM:2656934.2631917,mogavero2014behavioral},
conjecturing that large fragments of~\SL (subsuming~\ATL*) would have
\EXPTIME[2] model-checking algorithms with such limited
dependences.

%


\subparagraph{Our contributions.}
We~follow this line of work by performing a more thorough exploration
of strategy dependences in (a~fragment of)~\SL. We~mainly follow the
framework of~\cite{mogavero2014behavioral}, based on a kind of
Skolemization of the formula: for instance, a~formula of the form
$(\forall x_i\exists y_i)_i.\ \xi$ is satisfied if there exists
a \emph{dependence map}~$\theta$ defining each
existentially-quantified strategy~$y_j$ based on the
universally-quantified strategies~$(x_i)_i$. In~order to recover the
classical semantics of~\SL, it~is only required that the strategy
$\theta((x_i)_i)(y_j)$ (i.e. the strategy assigned to the
existentially-quantified variable~$y_j$ by~$\theta((x_i)_i)$) only
depends on $(x_i)_{i<j}$.

Based on this definition, other constraints can be imposed on
dependence maps, in order to refine the dependences of
existentially-quantified strategies on universally-quantified ones.
\emph{Elementary dependences}~\cite{mogavero2014behavioral} only
allows existentially-quantified strategy~$y_j$ to depend on the values
of~$(x_i)_{i<j}$ along the current history. This gives rise to two
different semantics in general, but fragments of~\SL have been defined
on which the classic and elementary semantics would
coincide~\cite{Mogavero:2012:MAD:2403555.2403575,SLCGandDG2013Faux}.

We~introduce yet another kind of dependences, which we coin
\emph{timeline dependences}, and which extends elementary dependences
by allowing existentially-quantified strategies to additionally depend
on \emph{all} universally-quantified strategies along \emph{strict
  prefixes} of the current history. This we believe is even more
relevant than elementary dependences. 


We~study and compare those three dependences (classic, elementary and
timeline), showing that they correspond to three distinct
semantics. Because the semantics based on dependence maps is defined
in terms of the \emph{existence} of a witness map, we~show that the
syntactic negation of a formula may not correspond to its semantic
negation: there are cases where both a formula~$\phi$ and its
syntactic negation~$\neg\phi$ fail to hold (i.e., none of them has a
witness map). This phenomenon is already present, but had not been
formally identified,
in~\cite{Mogavero:2014:RSM:2656934.2631917,mogavero2014behavioral}. The
main contribution of the present paper is the definition of a fragment
of~\SL for which syntactic and semantic negations coincide for the
timeline semantics. As~an (important) side result, we show that model
checking this fragment under the timeline semantics is
\EXPTIME[2]-complete.






  \section{Definitions}
  \label{se::Definitions}
  \subsection{Concurrent game structures}
\label{se::Definitions:::ss:CGS}
For the rest of this paper, we~fix a finite set~$\AP$ of atomic propositions,
a finite set~$\calV$ of variables, and a finite set~$\Agt$ of agents (or~players).

A \emph{concurrent game structure}
is a tuple $\calG = \tuple{\Act,\Q,\Delta,\labels}$
where $\Act$~is a
finite set of actions, $\Q$~is a finite set of states,
$\Delta\colon \Q \times\Act^\Agt\to \Q$ is the transition function,
and $\labels\colon \Q\to2^{\AP}$ is a labelling function. An~element
of $\Act^\Agt$ will be called a \emph{move vector}.
For~any~${q \in \Q}$, we let $\succs(q)$ be the set $\{q'\in
Q\mid \exists m\in\Act^\Agt.\ q'=\Delta(q,m)\}$. For~the sake of
simplicity, we~assume in the sequel that $\succs(q)\not=\varnothing$ for
any~$q\in Q$.
%
A~game~$\calG$ is said \emph{turn-based} whenever for every
state~$q \in \Q$, there is a player $\own(q)\in \Agt$ (named
the \emph{owner} of~$q$) such that for any two move vectors~$m_1$
and~$m_2$ with $m_1(\own(q))=m_2(\own(q))$, it~holds
$\Delta(q,m_1)=\Delta(q,m_2)$. 
Figure~\ref{fig-ex} displays an example of a (turn-based) game.


Fix a state~$q\in \Q$.  A~\emph{play} in~$\calG$ from~$q$ is an
infinite sequence $\pi=(q_i)_{i \in \bbN}$ of states in~$Q$ such that
$q_0=q$ and $q_i\in\succs(q_{i-1})$ for all~$i>0$.  We~write
$\Play_{\calG}(q)$
for the set of plays in~$\calG$ from~$q$.
In this and all similar notations, we
might omit to mention~$\calG$ when it is clear from the context,
and~$q$ when we consider the union over all~$q\in Q$.
A~(strict) prefix of a
play~$\pi$ is a finite sequence $\rho=(q_i)_{0\leq i\leq L}$, for
some~$L\in\bbN$. We~write $\Prefx(\pi)$ for the set of strict prefixes
of play~$\pi$. Such finite prefixes are called \emph{histories}, and
we~let $\Hist_{\calG}(q)=\Prefx(\Play_\calG(q))$. We~extend the notion of
strict prefixes and the notation~$\Prefx$ to histories in the natural
way, requiring in particular that $\rho\notin\Prefx(\rho)$. A~(finite)
extension of a history~$\rho$ is any history~$\rho'$ such that
$\rho\in\Prefx(\rho')$.
%
Let $\rho=(q_i)_{i\leq L}$ be a history. We~define $\fst\rho=q_0$ and
$\lst{\rho}=q_L$. Let $\rho'=(q'_j)_{j\leq L'}$ be a history
from~$\lst\rho$. The~\emph{concatenation} of~$\rho$ and~$\rho'$ is
then defined as the path $\rho\cdot\rho'=(q''_k)_{k\leq L+L'}$ such
that $q''_k=q_k$ when $k\leq L$ and $q''_k=q'_{k-L}$ when $L\geq k$
(notice that we required $q'_0=q_L$).



%


A~\emph{strategy} from~$q$ is a mapping
$\delta\colon \Hist_\calG(q) \to \Act$.  We~write~$\Strat_{\calG}(q)$
for the set of strategies in~$\calG$ from~$q$.  Given a
strategy~$\delta\in\Strat(q)$ and a history~$\rho$ from~$q$,
the~\emph{translation}~$\delta_{\overrightarrow \rho}$ of~$\delta$
by~$\rho$ is the strategy $\delta_{\overrightarrow \rho}$ from
$\lst\rho$ defined by
$\delta_{\overrightarrow \rho}(\rho')=\delta(\rho \cdot \rho')$ for
any $\rho'\in\Hist(\lst{\rho})$.
%
%
A~\emph{valuation} from~$q$ is a
partial function $\chi\colon \calV \cup \Agt\rightharpoonup\Strat(q)$.
As~usual, for
any partial function~$f$, we~write~$\dom{f}$ for the domain of~$f$.

Let $q\in Q$ and $\chi$ be a valuation
from~$q$. If~$\Agt\subseteq\dom\chi$, then $\chi$ induces a unique
play from~$q$, called its \emph{outcome}, and defined as
$\out(q,\chi)=(q_i)_{i\in\bbN}$ such that $q_0=q$ and for every~${i
\in \bbN}$, we~have $q_{i+1}=\Delta(q_i,m_{i})$ with 
$m_{i}(A)=\chi(A)((q_{j})_{j\leq i})$ for every~$A \in \Agt$.


\subsection{Strategy Logic with boolean goals}
\label{se::Definitions:::ss:SLBGdef}
Strategy Logic~(\SL for short) was introduced
in~\cite{CHATTERJEE2010677}, and further extended and studied
in~\cite{fsttcs2010-MMV,Mogavero:2014:RSM:2656934.2631917}, as a rich
logical formalism for expressing properties of games. \SL~manipulates
strategies as first-order elements, assigns them to players, and
expresses \LTL properties on the outcomes of the resulting strategic
interactions.
This results in a very expressive temporal logic, for which
satisfiability is undecidable~\cite{fsttcs2010-MMV} and model checking
is \TOWER-complete~\cite{Mogavero:2014:RSM:2656934.2631917,BGM-fsttcs15}.
%
In~this paper, we focus on
a restricted fragment of~\SL, called~\SLBGf
(where \textsf{BG} stands for \emph{boolean
  goals}~\cite{Mogavero:2014:RSM:2656934.2631917}, and the
symbol~$\flat$ indicates that we do not allow nesting of (closed)
subformulas; we~discuss this latter restriction below).




\subparagraph{Syntax.}
Formulas in~\SLBGf are built along the following grammar
\begin{xalignat*}2
\SLBGf \ni \phi & \coloncolonequals \exists x.\ \phi \mid \forall x.\ \phi \mid \xi
  & 
  \xi & \coloncolonequals \non\xi\mid \xi\wedge\xi \mid \xi\vee\xi \mid \beta 
  \\
  \beta & \coloncolonequals \assign{\sigma}.\ \psi
  & 
  \psi & \coloncolonequals \neg\psi \mid \psi\ou\psi \mid \psi\et\psi\mid \X\psi
  \mid \psi\Until\psi \mid p
\end{xalignat*}
where $x$ ranges over~$\calV$, $\sigma$~ranges over the set~$\calV^\Agt$
of \emph{full assignments}, and
$p$ ranges over~$\AP$.
A~\emph{goal} is a formula of the form~$\beta$ in the grammar above;
it~expresses an \LTL property~$\psi$ on the outcome of the
mapping~$\sigma$. Formulas in~\SLBGf are thus made of an initial block
of first-order quantifiers (selecting strategies for variables
in~$\calV$), followed by a boolean combination of such goals.

\subparagraph{Free variables.}
With any subformula~$\zeta$ of some formula~$\phi\in\SLBGf$,
we~associate its~set of \emph{free agents and variables}, which we
write~$\free \zeta$. It~contains the agents and variables that have to
be associated with a strategy in order to unequivocally evaluate~$\zeta$
(as~will be seen from the definition of the semantics of \SLBGf
below). The~set $\free\zeta$ is defined inductively:
\begin{xalignat*}2
  \free p &= \varnothing \quad \text{for all $p\in\AP$}\qquad &
  \free{\X\psi} &= \Agt\cup\free{\psi} \\
  \noalign{\pagebreak[2]}
  \free{\neg\alpha} &= \free\alpha &
  \free{\psi_1\Until\psi_2} &= \Agt\cup\free{\psi_1}\cup \free{\psi_2} \\
  \noalign{\pagebreak[2]}
  \free{\alpha_1\ou\alpha_2} &= \free{\alpha_1} \cup \free{\alpha_2} &
  \free{\exists x.\ \phi} &= \free\phi \setminus\{x\} \\
  \free{\alpha_1\et\alpha_2} &= \free{\alpha_1} \cup \free{\alpha_2} &
  \free{\forall x.\ \phi} &= \free\phi \setminus\{x\} \\
    \free{\assign{\sigma}.\ \phi} &
    \multicolumn{3}{l}{${} =(\free\phi\cup \sigma(\Agt\cap\free\phi))\setminus \Agt
    $}
\end{xalignat*}
Subformula $\zeta$ is said to be \emph{closed} whenever
$\free{\zeta}=\varnothing$. 
We~can now comment on our choice of considering the flat fragment of
\SLBG: the~full fragment, as defined
in~\cite{Mogavero:2014:RSM:2656934.2631917}, allows for nesting
\emph{closed} \SLBG formulas in place of atomic propositions.  The
meaning of such nesting in our setting is ambiguous, because our
semantics (in~Sections~\ref{sec-depend} to~\ref{sec-sleg}) are defined
in terms of the \emph{existence of a witness}, which does not easily
propagate in formulas.
In~particular, as~we explain later in the paper, the semantics of the
negation of a formula (there exist a witness for~$\non\phi$) does not
coincide with the negation of the semantics (there is no witness
for~$\phi$);
thus substituting a subformula and substituting its negation may
return different results.






\subparagraph{Semantics.}
Fix a state~$q\in Q$, and a valuation
$\chi\colon \calV\cup\Agt\to \Strat(q)$.  We~inductively define the
semantics of a subformula~$\alpha$ of a formula of \SLBGf at~$q$ under
valuation~$\chi$, requiring $\free\alpha\subseteq \dom\chi$.
We~omit the easy cases of boolean combinations and atomic propositions.

Given a mapping $\sigma\colon \Agt\to\calV$, the semantics of strategy
assignments is defined as follows:
\[
\calG,q\models_{\chi} \assign{\sigma}.\ \psi
\quad \iff \quad 
\calG,q\models_{\chi[A\in\Agt\mapsto\chi(\sigma(A))]} \psi.
\]
Notice that, writing
$\chi'=\chi[A\in\Agt\mapsto\chi(\sigma(A))]$, we~have
$\free\psi\subseteq \dom{\chi'}$ if $\free\alpha\subseteq\dom\chi$, so
that our inductive definition is sound.

We~now consider path formulas $\psi=\X\psi_1$ and
$\psi=\psi_1\Until\psi_2$.  Since
$\Agt\subseteq\free\psi\subseteq \dom\chi$, the~valuation~$\chi$
induces a unique outcome $\out(q,\chi)=(q_i)_{i\in\bbN}$ from~$q$.
For~$n\in\bbN$, we~write $\out_n (q,\chi)=(q_i)_{i\leq n}$, and define
$\chi_{\overrightarrow n}$ as the valuation obtained by shifting all
the strategies in the image of~$\chi$ by~$\out_n (q,\chi)$.  Under the
same conditions, we~also define $q_{\overrightarrow n}=\lst{\out_n
(q,\chi)}$.  We~then~set
\begin{alignat*}1
  \calG,q\models_{\chi} \X\psi_1
  &\quad\iff\quad 
  \calG, q_{\overrightarrow{1}} \models_{\chi_{\overrightarrow 1}} \psi_1
  \\
  \calG,q\models_{\chi} \psi_1\Until\psi_2
  &\quad\iff\quad 
  \exists k\in\bbN.\ \calG, q_{\overrightarrow{k}} \models_{\chi_{\overrightarrow k}} \psi_2
  \quad \text{and} \quad 
  \forall 0\leq j< k.\ \calG, q_{\overrightarrow{j}} \models_{\chi_{\overrightarrow j}} \psi_1.
\end{alignat*}

It remains to define the semantics of the strategy quantifiers. This
is actually what this paper is all about. We~provide here the original
semantics, and discuss alternatives in the following sections:
\begin{align*}
  \calG,q \models_{\chi} \exists x.\phi 
  \quad\iff\quad 
  \exists \delta\in\Strat(q).\ \calG,q\models_{\chi[x\mapsto \delta]} \phi.
\end{align*}
In the sequel, we use classical shorthands, such as
$\top$ for $p \ou \non p$ (for any~$p\in\AP$), 
$\F\psi$ for $\top\Until\psi$ (\emph{eventually}~$\psi$), and $\G\psi$ for $\non\F\non\psi$ (\emph{always}~$\psi$). 


\begin{example}
  \label{first-example} 
  We consider the (turn-based) game~$\calG$ is depicted on
  Fig.~\ref{jhvgjjgjh}. We~name the players after the shape of the
  state they control.
  The \SLBG formula~$\phi$ to the right of
  Fig.~\ref{jhvgjjgjh} has four quantified variables and
  two goals.
  We~show that this formula evaluates to~true at~$q_0$:
  fix a strategy~$\delta_y$ (to be played by
  player~\CircleFill[\FigColA]{}); because~$\calG$ is turn-based,
  we~identify the actions of the owner of a state with the resulting
  target state, so that $\delta_y(q_0q_1)$ will be either~$p_1$
  or~$p_2$. We~then define strategy~$\delta_z$ (to be played
  by~\DiamondFill[\FigColC]{}) as
  $\delta_z(q_0q_2)=\delta_y(q_0q_1)$. Then~clearly, for~any strategy
  assigned to player~\BoxFill[\FigColB]{}, one of the goals of
  formula~$\phi$ holds true, so that~$\phi$ itself evaluates to true.
%
\begin{figure}[t]
\centering
    \begin{tikzpicture} [xscale=1, yscale=.7, every node/.style={draw}]
      \node [rectangle, minimum size=6mm,draw=\FigColB, fill=\FigColB!40!white] (root) at (0,0) {$q_0$};
      \node [circle, inner sep=0pt, minimum size=7mm,draw=\FigColA, fill=\FigColA!40!white] (sqr) at (1,-0.75) {$q_1$};
      \node [diamond, inner sep=0pt, minimum size=8mm,draw=\FigColC, fill=\FigColC!40!white] (diam) at (1,0.75) {$q_2$};
      \node [draw=none] (A) at (2.2,-0.75) {$\ p_1$};
      \node [draw=none] (B) at (2.2,0.75) {$\ p_2$};
      \draw [-latex'] (root) to (sqr); \draw [-latex'] (root) to (diam);
      \draw [-latex'] (sqr) to (A); \draw [-latex'] (sqr) to (B);
      \draw [-latex'] (diam) to (B); \draw [-latex'] (diam) to (A);
   \path[use as bounding box] (13,0);
   \path(8.4,0.3) node[draw=none]{\begin{minipage}{.85\linewidth}
   \begin{alignat*}1
      \phi = \forall y.\exists z.\forall x_A.\forall x_B.\ \bigvee
      \begin{cases}
	\assign{\BoxFill[\FigColB]{}\mapsto x_A; \CircleFill[\FigColA]{}\mapsto y;
          \DiamondFill[\FigColC]{} \mapsto z}.\ \F p_1
	\\
	\assign{\BoxFill[\FigColB]{}\mapsto x_B; \CircleFill[\FigColA]{}\mapsto y;
          \DiamondFill[\FigColC]{} \mapsto z}.\ \F p_2
      \end{cases}
    \end{alignat*}
    \end{minipage}};
  \end{tikzpicture}
\caption{A game and a \protect\SLBG formula.}
  \label{jhvgjjgjh}\label{fig-ex}
\end{figure}
\end{example}


\subparagraph{Subclasses of \protect\SLBG.}
Because of the high complexity and subtlety of reasoning with~\SL and \SLBG, 
several restrictions of \SLBG have been considered in the
literature~\cite{Mogavero:2012:MAD:2403555.2403575,SLCGandDG2013Faux,mogavero2014behavioral},
by adding further restrictions to boolean combinations in the grammar
defining the syntax:
\begin{itemize}
\item \SLOneG restricts \SLBG to a unique goal. 
  \SLOneGf is then defined from the grammar of \SLBGf by setting
  $\xi \coloncolonequals \beta$ in the grammar;
\item the larger fragment \SLCG allows for \emph{conjunctions} of goals.
  \SLCGf~corresponds to formulas defined with
  $\xi \coloncolonequals \xi\et\xi \mid \beta$;
\item similarly, \SLDG only allows \emph{disjunctions} of goals, i.e.
  $\xi \coloncolonequals \xi\vee \xi \mid \beta$;
\item finally, \SLAG mixes conjunctions and disjunctions in a
  restricted way. Goals in~\SLAGf can be combined using the following grammar:
  $\xi \coloncolonequals \beta\et\xi \mid  \beta\ou\xi \mid\beta$.
\end{itemize}

\noindent 
In the sequel, we write a generic \SLBGf formula $\phi$ as
  $(Q_i x_i)_{1\leq i \le l} .\ \xi(\beta_j.\ \psi_j)_{j \le n}$ where:
    \begin{itemize}
    \item $(Q_i x_i)_{i \le l}$ is a block of quantifications, with
      $\{x_i \mid 1 \le i \le l\}\subseteq \calV$
      and $Q_i\in\{\exists,  \forall\}$, for every $1 \le i \le l$;
    \item $\xi(g_1,...,g_n)$ is a boolean combination of its arguments;
    \item for all $1 \le j \le n$,  $\beta_j.\ \psi_j$ is a goal:
      $\beta_j$ is a full assignment and $\psi_j$ is an \LTL formula.
    \end{itemize}

  \section{Strategy dependences}
  \label{se::Framework}\label{sec-depend}

We~now follow the framework
of~\cite{
Mogavero:2014:RSM:2656934.2631917,
mogavero2014behavioral}
and define
the semantics of \SLBGf in terms of \emph{dependence maps}.  This
approach provides a fine way of controlling how
\emph{existentially-quantified} strategies depend on previously
selected strategies (in a quantifier block). Considering again
Example~\ref{first-example}, we~notice that the value of the
existentially-quantified strategy~$\delta_z$ after history~$q_0q_2$
depends on the value of strategy~$\delta_y$ on history~$q_0q_1$, which
may not be realistic.
Using dependence maps, we can limit such dependences.



\subparagraph{Dependence maps.}
%
Consider an \SLBGf formula $\phi = (Q_i x_i)_{1\leq i \le l} .\ \xi(\beta_j.\
\varphi_j)_{j \le n}$, assuming w.l.o.g. that
$\{ x_i \mid 1\leq i\leq l\}=\calV$. 
We~let $\calV^\forall =\{ x_i \mid Q_i=\mathord\forall \}\subseteq\calV$
be the set of universally-quantified variables of~$\phi$.
A function~$\theta\colon \Strat^{\calV^\forall} \to \Strat^{\calV}$
is a \emph{$\phi$-map} (or~\emph{map}
when $\phi$ is clear from the context) if
$\theta(w)(x_i)(\rho)=w(x_i)(\rho)$ for any
$w \in \Strat^{\calV^\forall}$, any~$x_i\in\calV^\forall$, and any
history~$\rho$. In~other words, $\theta(w)$ extends~$w$ to~\calV.
This general notion allows any existentially-quantified variable to
depend on \emph{all} universally-quantified ones (dependence on
existentially-quantified variables is implicit: all
existentially-quantified variables are assigned through a single map,
hence they all depend on the others); we~add further restrictions
later on.
Using maps, we~may then define new semantics for~\SLBGf:
generally speaking, formula $\phi = (Q_i x_i)_{1\leq i \le l} .\
\xi(\beta_j.\ \varphi_j)_{j \le n}$ holds true if there exists a
$\phi$-map~$\theta$
such that, for any $w\colon \calV^\forall \to \Strat$, the valuation
$\theta(w)$ makes $\xi(\beta_j.\ \varphi_j)_{j \le n}$ hold true.

\emph{Classic maps} are dependence maps in which the order of
quantification is respected:
\begin{multline}
\forall w_1,w_2\in\Strat^{\calV^\forall}.\
\forall x_i\in\calV\setminus\calV^\forall.\
\\
\bigl(\forall x_j\in \calV^\forall\cap\{x_j\mid j<i\}.\ w_1(x_j)=w_2(x_j)\bigr)
\Rightarrow
\bigl(\theta(w_1)(x_i)=\theta(w_2)(x_i)\bigr).
\tag{C}
\label{eq-C}
\end{multline}
In words, if $w_1$ and~$w_2$ coincide on
$\calV^\forall\cap\{x_j\mid j<i\}$, then $\theta(w_1)$ and
$\theta(w_2)$ coincide on~$x_i$.

\emph{Elementary
maps}~\cite{Mogavero:2014:RSM:2656934.2631917,
Mogavero:2012:MAD:2403555.2403575} have to satisfy a more
restrictive condition: for those maps, the value of an
existentially-quantified strategy at any history~$\rho$ may only
depend on the value of earlier universally-quantified
strategies \emph{along~$\rho$}. 
This may be written as:
\begin{multline}
\forall w_1,w_2\in\Strat^{\calV^\forall}.\
\forall x_i\in\calV\setminus\calV^\forall.\
\forall \rho\in\Hist.\ 
\\
\bigl(\forall x_j\in \calV^\forall\cap\{x_k\mid k<i\}.\
\forall \rho'\in\Prefx(\rho)\cup\{\rho\}.\ 
w_1(x_j)(\rho')=w_2(x_j)(\rho')\bigr)
\Rightarrow\\
\bigl(\theta(w_1)(x_i)(\rho)=\theta(w_2)(x_i)(\rho)\bigl).
\tag{E}
\label{eq-E}
\end{multline}
In this case, for any history~$\rho$,
if two valuations $w_1$ and $w_2$ of the universally-quantified variables
coincide on the variables quantified before~$x_i$ all along~$\rho$, 
then $\theta(w_1)(x_i)$ and $\theta(w_2)(x_i)$ have to coincide at~$\rho$.

The difference between both kinds of dependences is illustrated on
Fig.~\ref{fig-dep}: for classic maps, the existentially-quantified
strategy~$x_2$ may depend on the whole strategy~$x_1$, while it may
only depend on the value of~$x_1$ along the current history for
elementary maps. Notice that a map satisfying~\eqref{eq-E} also
satisfies~\eqref{eq-C}.

\def\mytree#1#2#3{%
  \draw (0,0) node[moyrond,vert,inner sep=2pt] (#1a) {$#2$};
  \path (#1a) node[coordinate] (#1A) {};
  \draw (#1a.-35) -- +(-72:2.2cm);
  \draw (#1a.-145) -- +(-108:2.2cm);
  \draw (-70:1cm) node[minicarre] (#1b) {};
  \draw (-110:1cm) node[minicarre] (#1b') {};
  \path (#1b) node[coordinate] (#1B) {};
  \path (#1b') node[coordinate] (#1B') {};
  \path (#1b) -- +(-77:8mm) node[minicarre] (#1c) {};
  \path (#1b) -- +(-103:8mm) node[minicarre,#3] (#1c') {};
  \path (#1b') -- +(-77:8mm) node[minicarre] (#1c'') {};
  \path (#1b') -- +(-103:8mm) node[minicarre] (#1c''') {};
  \draw (#1a) -- (#1b);
  \draw (#1a) -- (#1b');
  \draw (#1b) -- (#1c);
  \draw (#1b) -- (#1c');
  \draw (#1b') -- (#1c'');
  \draw (#1b') -- (#1c''');
  \foreach \i in {c,c',c'',c'''}
    {\draw[dashed] (#1\i) -- +(-80:6mm);\draw[dashed] (#1\i) -- +(-100:6mm);}
}

\begin{figure}[b]
  \centering
  \begin{tikzpicture}
    \begin{scope}
    \begin{scope}
      \mytree{1}{\forall x_1}{line width=.5mm}
      \fill[opacity=.4,orange] (1a.-35) -- ++(-72:2.2cm) -- ++(-2,0) --
        (1a.-145) ++ (-108:22mm) -- ++(2,0);
      \mytree{1}{\forall x_1}{line width=.5mm}
    \end{scope}
    \begin{scope}[xshift=2.2cm]
      \mytree{2}{\exists x_2}{line width=.5mm}
    \end{scope}
    \begin{scope}[xshift=4.4cm]
      \mytree{3}{\forall x_3}{line width=.5mm}
    \end{scope}
    \draw[dotted,line width=.5mm] (0,-.6) edge[bend left,-latex'] (2c');
    \end{scope}
    \begin{scope}[xshift=7.5cm]
    \begin{scope}
      \mytree{1}{\forall x_1}{line width=.5mm}
      \draw[opacity=.4,line width=2mm,orange] (1A) -- (1B)
        node[pos=.8,inner sep=0pt] (1d) {} -- (1c');
      \mytree{1}{\forall x_1}{line width=.5mm}
    \end{scope}
    \begin{scope}[xshift=2.2cm]
      \mytree{2}{\exists x_2}{line width=.5mm}
    \end{scope}
    \begin{scope}[xshift=4.4cm]
      \mytree{3}{\forall x_3}{line width=.5mm}
    \end{scope}
    \draw[dotted,line width=.5mm] (1d) edge[bend left,-latex'] (2c');
    \end{scope}
  \end{tikzpicture}
  \caption{Classical (left) vs elementary (right) dependences for a formula $\forall x_1.\ \exists x_2.\ \forall x_3.\ \xi$}
  \label{fig-dep}
\end{figure}
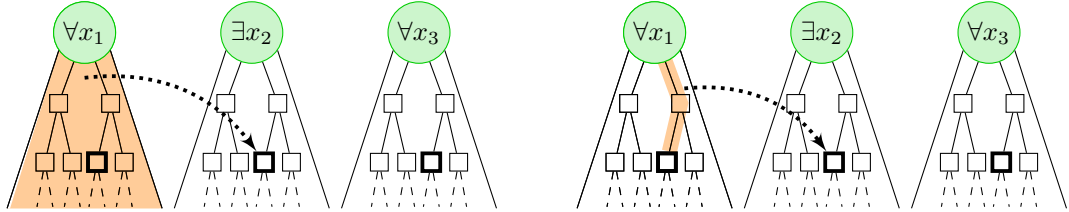

\subparagraph{Satisfaction relations.}
%
Pick a formula $\phi=(Q_i x_i)_{1\leq i\leq l} .\ \xi \bigl( \beta_j .\
\formuleLTL_j \bigr)_{j\leq n}$ in~\SLBGf. We~define:
\begin{alignat*}1
  \calG, q \models^C 
  \phi \quad \textup{iff} \quad \exists \theta\text{ satisfying~\eqref{eq-C}}.\
  \forall w \in
  \Strat^{\calV^\forall } 
  .\ \calG,q\models_{\theta(w)} \xi \big( \beta_j \formuleLTL_j
  \big)_{j\leq n}
\end{alignat*}
As explained above, this actually corresponds to the usual semantics
of \SLBGf as~given in
Section~\ref{se::Definitions}~\cite[Theorem~4.6]{Mogavero:2014:RSM:2656934.2631917}.
When $\calG,q\models^C\phi$, a~map~$\theta$ satisfying the conditions above is
called a $\ref{eq-C}$-\emph{witness} of~$\phi$ for~$\calG$ and~$q$.
%
Similarly, we~define the \emph{elementary
semantics}~\cite{Mogavero:2014:RSM:2656934.2631917} as:
\begin{alignat*}1
  \calG, q \models^E 
  \phi \quad \textup{iff} \quad \exists \theta\text{ satisfying~\eqref{eq-E}}.\
  \forall w \in
  \Strat^{\calV^\forall } 
  .\ \calG,q\models_{\theta(w)} \xi \big( \beta_j \formuleLTL_j
  \big)_{j\leq n}
\end{alignat*}
Again, when such a map exists, it~is called an $\ref{eq-E}$-witness.
Notice that since Property~\eqref{eq-E} implies Property~\eqref{eq-C},
we~have ${\calG,q\models^E \phi} \impl {\calG,q\models^C \phi}$ for
any ${\phi\in\SLBGf}$. 
%
This corresponds to the intuition that it is harder to satisfy a \SLBGf
formula when dependences are more restricted. 
The contrapositive statement then
raises questions about the negation of formulas. 

\subparagraph{The syntactic vs.\ semantic negations.}
If $\phi = (Q_i x_i) _{1\leq i \le l} \xi(\beta_j \varphi_j)_{j \le n}$ is
an \SLBGf formula, its syntactic negation~$\neg \phi$ is the formula
$(\overline{Q}_i x_i)_{i \le l} (\neg\xi)(\beta_j \varphi_j)_{j \le
  n}$, where $\overline{Q}_i =\exists$ if $Q_i =\forall$ and
$\overline{Q}_i = \forall$ if $Q_i =\exists$.  Looking at the
definitions of~$\models^C$ and $\models^E$, it~could be the case that
e.g. $\calG,q\models^{C} \phi$ and $\calG,q\models^{C}
\neg\phi$: this only requires the existence of two adequate maps.
However, since $\models^{C}$ and~$\models$ coincide, and since
$\calG,q\models \phi \iff \calG,q\not\models \neg\phi$ in the usual semantics,
we~get 
$\calG,q\models^{C} \phi \iff \calG,q\not\models^{C} \neg\phi$.  Also, since
$\calG,q\models^E \phi \impl \calG,q\models^C \phi$, we~also get
$\calG,q\models^{E} \phi \impl \calG,q\not\models^{E} \neg\phi$.
As we now show, the converse implication may fail to hold.


\begin{restatable}{proposition}{PropositionTroisAvril}
  \label{Chapitre3Framework::Prop::NeitherHold::31MArs2017}\label{prop-neg}
  There exist a (one-player) game $\calG$
  with initial state $q_0$ and a formula $\phi \in \SLBGf$ such that
  $\calG,q_0 \not\models^{E} \phi$ and
  $\calG,q_0 \not\models^{E} \neg \phi$.
\end{restatable}

\begin{proof}
  Consider the formula and
  the one-player game of Fig.~\ref{fig-negnotneg}.
  \begin{figure}[tb]
   \begin{tikzpicture}[xscale=0.75, yscale=0.8]
    \node [rond6,fill=\FigColA!40!white] (qinit) at (0,0) {$q_0$};
    \node [rond6, fill=\FigColA!40!white] (a) at (-1.5,-1) {$A$}; \node (p1) at (-1.5-0.75,-2) {$p_1$}; \node (p2) at (-1.5+0.75,-2) {$p_2$};
    \node [rond6, fill=\FigColA!40!white] (b) at (1.5,-1) {$B$}; \node (p3) at (1.5-0.75,-2) {$p_1$}; \node (p4) at (1.5+0.75,-2) {$p_2$};
    \draw[-latex'] (qinit) to (a); \draw[-latex'] (qinit) to (b);
    \draw[-latex'] (a) to (p1);\draw[-latex'] (a) to (p2);
    \draw[-latex'] (b) to (p3);\draw[-latex'] (b) to (p4);
   \path[use as bounding box] (15,0);
\path (9.5,-.2) node{%
  \begin{minipage}[t][][b]{0.8\textwidth}
   \begin{alignat*}1
    \phi=\forall x. \exists y.\ \ET
    \begin{cases}
      \assign{\CircleFill[\FigColA]{}\mapsto y}.~ \F  B
        \\
      \assign{\CircleFill[\FigColA]{}\mapsto x}.\ \F p_1  \iff
      \assign{\CircleFill[\FigColA]{}\mapsto y}.\ \F p_1
    \end{cases}
  \end{alignat*}
  \end{minipage}};
  \end{tikzpicture}
  \caption{A~game~$\calG$ and an \SLBGf formula~$\phi$ such that
  $\calG,q_0 \not\models^{E} \phi$ and $\calG,q_0
    \not\models^{E} \neg \phi$.}\label{fig-negnotneg}
  \end{figure}
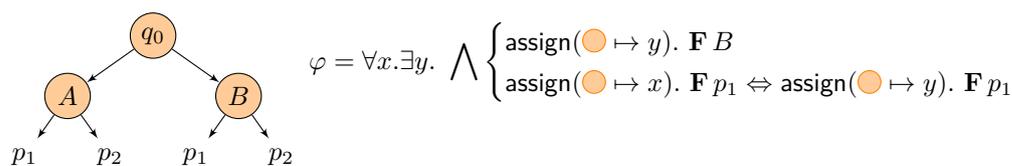
  We start by proving that $\calG,q_0 \not\models^{E} \phi$, by
  looking for a witness for~$\phi$.
  First, for the first goal in the conjunction to be fulfilled, the
  strategy assigned to~$y$ must play to~$B$ from~$q_0$, whatever the
  valuation~$w$ for the universal variable~$x$.  We~abbreviate this as
  $\theta(w)(y)(q_0)=B$ in the sequel.  Now, for a valuation~$w$
  s.t. $w(x)(q_0)=A$, we~must have $\theta(w)(y)(q_0\cdot
  B)=w(x)(q_0\cdot A)$ in order to fulfill the second goal.  Such
  dependences are not allowed in the elementary semantics.
  
  We now prove that
  $\calG,q_0 \not\models^{E} \neg \phi$. Indeed, following the
  previous discussion, we~easily get that
  $\calG,q_0 \models^{C} \phi$, by letting $\theta(w)(y)(q_0)=B$ and
  $\theta(w)(y)(q_0\cdot B)=w(x)(q_0\cdot A)$ if $w(x)(q_0)=A$, and
  $\theta(w)(y)(q_0\cdot B)=w(x)(q_0\cdot B)$ if 
  $w(x)(q_0)=B$. As~explained above, this entails
  $\calG,q_0 \not\models^{C} \neg \phi$, and
  $\calG,q_0 \not\models^{E} \neg \phi$.
\end{proof}

The case of~\SLOneGf is simpler and we get:
\begin{restatable}{proposition}{PropositionquatreAvril}
  \label{prop-negSL1G}
  For any game~$\calG$
  with initial state $q_0$, and any formula $\phi \in \SLOneGf$,
  it~holds
  $\calG,q_0 \models^{E} \phi  \iff
  \calG,q_0 \not\models^{E} \neg \phi$.
\end{restatable}


\begin{proof}[Sketch of proof]
%
        This result follows
        from~\cite[Corollary~4.21]{Mogavero:2014:RSM:2656934.2631917},
        which states that $\models^C$ and $\models^E$ coincide
        on~\SLOneG. Because it is central in our approach,
        we~sketch a direct proof here (with a full proof
        in Appendix~\ref{app-Prop4}), using similar ingredients:
        it~consists in encoding the problem whether
        $\calG,q_0 \models^{E} \phi$ into a two-player turn-based game
        with a parity-winning objective.

  The construction is as follows:
  the interaction between
  existential and universal quantifications of the formula is
  integrated into the game structure, replacing each state of~$\calG$
  with a tree-shaped subgame where 
  Player~$P_\exists$ selects existentially-quantified actions and
  Player~$P_\forall$ selects universally-quantified ones.  The unique
  goal of the formula is then incorporated into the game via a
  deterministic parity automaton, yielding a two-player turn-based
  parity game. We~then show that $\calG,q_0\models^{E}\phi$ if, and
  only~if, Player~$P_\exists$ has a winning strategy in the resulting
  turn-based parity game, while $\calG,q_0\models^{E}\neg\phi$ if, and
  only~if, Player~$P_\forall$ has a winning strategy.  Those
  equivalences hold for the elementary semantics because memoryless
  strategies are sufficient in parity games.
  Proposition~\ref{prop-negSL1G} then follows by determinacy of those
  games~\cite{EJ91,Mos91}.
\end{proof}

\looseness=-1
Note that the construction of the parity game gives an effective
algorithm for the model-checking problem of \SLOneGf, which runs in
time doubly-exponential in the size of the formula, and polynomial in
the size of the game structure; we~recover the result
of~\cite{Mogavero:2014:RSM:2656934.2631917} for that problem.

\subparagraph{Comparison of $\models^C$ and~$\models^E$.}  
A~consequence of Prop.~\ref{prop-negSL1G} is that $\models^C$ and
$\models^E$ coincide on~\SLOneGf (Corollary~4.21
of~\cite{Mogavero:2014:RSM:2656934.2631917}).  However, when
considering larger fragments, the~satisfaction relations are distinct (see the proof of Prop.~\ref{Chapitre3Framework::Prop::NeitherHold::31MArs2017} for a
candidate formula in~\SLCGf):%
%
\begin{restatable}{proposition}{PropSLCGcodeAIAHIFHEFE}
  \label{PropSLCGSideoupasSideRemoveDepende}\label{prop-CEdiff}
The relations $\models^C$ and $\models^E$ differ on~\SLCGf, 
as~well as on~\SLDGf.
\end{restatable}

\begin{remark}
  Proposition~\ref{PropSLCGSideoupasSideRemoveDepende} contradicts the
  claim in~\cite{SLCGandDG2013Faux} that $\models^{E}$ and
  $\models^{C}$ coincide on \SLCG (and~\SLDG).
  Indeed, in~\cite{SLCGandDG2013Faux}, the satisfaction relation for
  \SLDG and \SLCG is encoded into a two-player game in pretty much the
  same way as we did in the proof of Prop.~\ref{prop-negSL1G}.
  While
  this indeed rules out dependences outside the current history,
it~also gives information to Player~$P_{\exists}$ about the values
(over prefixes of the current history) of strategies that are
universally-quantified later in the quantification block.  This proof
technique works with~\SLOneGf because the single goal can be encoded
as a parity objective, for which memoryless strategies exist, so that
the extra information is not crucial.  In~the next section,
we~investigate the role of this extra information for larger fragments
of~\SLBGf.

\end{remark}



  \section{Timeline dependences}
  \label{se::AddingAction}\label{sec-timeline}
\looseness=-1
Following the discussion above, we~introduce a new type of dependences
between strategies (which we call \emph{timeline
dependences}). They~allow strategies to also observe (and~depend~on)
\emph{all} other universally-quantified strategies on the strict
prefix of the current history. For instance, for a block of
quantifiers~$\forall x_1.\ \exists x_2.\ \forall x_3$, the value
of~$x_2$ after history~$\rho$ may depend on the value of~$x_1$
on~$\rho$ and its prefixes (as for elementary maps), but also on the
value of~$x_3$ on the (strict) prefixes of~$\rho$. Such dependences
are depicted on Fig.~\ref{fig-semET}.  We~believe that such
dependences are relevant in many situations, especially for reactive
synthesis, since in this framework strategies really base their
decisions on what they could observe along the current history.

\def\mytree#1#2#3{%
  \draw (0,0) node[moyrond,vert,inner sep=2pt] (#1a) {$#2$};
  \path (#1a) node[coordinate] (#1A) {};
  \draw (#1a.-35) -- +(-72:2.2cm);
  \draw (#1a.-145) -- +(-108:2.2cm);
  \draw (-70:1cm) node[minicarre] (#1b) {};
  \draw (-110:1cm) node[minicarre] (#1b') {};
  \path (#1b) node[coordinate] (#1B) {};
  \path (#1b') node[coordinate] (#1B') {};
  \path (#1b) -- +(-77:8mm) node[minicarre] (#1c) {};
  \path (#1b) -- +(-103:8mm) node[minicarre,#3] (#1c') {};
  \path (#1b') -- +(-77:8mm) node[minicarre] (#1c'') {};
  \path (#1b') -- +(-103:8mm) node[minicarre] (#1c''') {};
  \draw (#1a) -- (#1b);
  \draw (#1a) -- (#1b');
  \draw (#1b) -- (#1c);
  \draw (#1b) -- (#1c');
  \draw (#1b') -- (#1c'');
  \draw (#1b') -- (#1c''');
  \foreach \i in {c,c',c'',c'''}
    {\draw[dashed] (#1\i) -- +(-80:6mm);\draw[dashed] (#1\i) -- +(-100:6mm);}
}

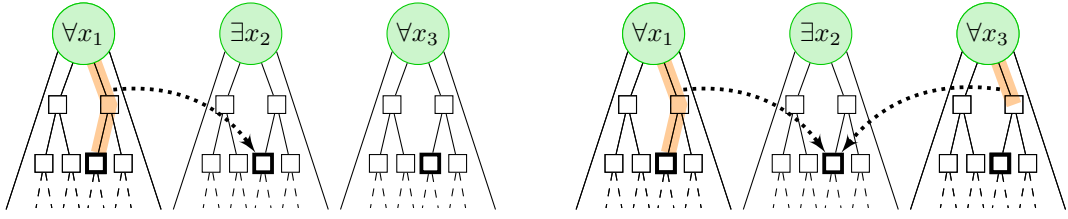
\begin{figure}[t]
  \centering
  \begin{tikzpicture}
    \begin{scope}
    \begin{scope}
      \mytree{1}{\forall x_1}{line width=.5mm}
      \draw[opacity=.4,line width=2mm,orange] (1A) -- (1B)
        node[pos=.8,inner sep=0pt] (1d) {} -- (1c');
      \mytree{1}{\forall x_1}{line width=.5mm}
    \end{scope}
    \begin{scope}[xshift=2.2cm]
      \mytree{2}{\exists x_2}{line width=.5mm}
    \end{scope}
    \begin{scope}[xshift=4.4cm]
      \mytree{3}{\forall x_3}{line width=.5mm}
    \end{scope}
    \draw[dotted,line width=.5mm] (1d) edge[bend left,-latex'] (2c');
    \end{scope}
    \begin{scope}[xshift=7.5cm]
    \begin{scope}
      \mytree{1}{\forall x_1}{line width=.5mm}
      \draw[opacity=.4,line width=2mm,orange] (1A) -- (1B)
        node[pos=.8,inner sep=0pt] (1d) {} -- (1c');
      \mytree{1}{\forall x_1}{line width=.5mm}
    \end{scope}
    \begin{scope}[xshift=2.2cm]
      \mytree{2}{\exists x_2}{line width=.5mm}
    \end{scope}
    \begin{scope}[xshift=4.4cm]
      \mytree{3}{\forall x_3}{line width=.5mm}
      \draw[opacity=.4,line width=2mm,orange] (3A) -- (3B)
        node[pos=.8,inner sep=0pt] (3d) {};
        \mytree{3}{\forall x_3}{line width=.5mm}
    \end{scope}
    \draw[dotted,line width=.5mm] (1d) edge[bend left,-latex'] (2c');   
    \draw[dotted,line width=.5mm] (3d) edge[bend right,-latex'] (2c');
    \end{scope}
  \end{tikzpicture}
  \caption{Elementary (left) vs timeline (right) dependences for a formula $\forall x_1.\ \exists x_2.\ \forall x_3.\ \xi$}
  \label{fig-semET}
\end{figure}




Formally, a~map~$\theta$ is a~\emph{timeline
map} if it satisfies the following condition:
\begin{multline}
\forall w_1,w_2\in\Strat^{\calV^\forall}.\
\forall x_i\in\calV\setminus\calV^\forall.\
\forall \rho\in\Hist.\ 
\\
\left(\begin{array}{r}
\forall x_j\in \calV^\forall\cap\{x_k\mid k<i\}.\
\forall \rho'\in\Prefx(\rho)\cup\{\rho\}.\ 
w_1(x_j)(\rho)=w_2(x_j)(\rho)
 \\
{} \et \forall x_j\in \calV^\forall.\
\forall \rho'\in\Prefx(\rho).\ 
w_1(x_j)(\rho)=w_2(x_j)(\rho)
\end{array}
\right)
\Rightarrow\\
\bigl(\theta(w_1)(x_i)(\rho)=\theta(w_2)(x_i)(\rho)\bigl).
\tag{T}
\label{eq-T}
\end{multline}

Using those maps, we~introduce the \emph{timeline semantics} of~\SLBGf:
\[
\calG,q\models^T \phi \quad\textup{iff}\quad
 \exists \theta\text{ satisfying~\eqref{eq-T}}.\
  \forall w \in
  \Strat^{\calV^\forall }.\
  \calG,q\models_{\theta(w)} \xi \big( \beta_j \formuleLTL_j \big)_{j\leq n}
\]
Such a map, if~any, is called a \ref{eq-T}-witness of~$\phi$
for~$\calG$ and~$q$. 
As in the previous section,
since Property~\eqref{eq-E} implies Property~\eqref{eq-T}, we~get that
an~\ref{eq-E}-witness is also a~\ref{eq-T}-witness, so that
$\calG,q\models^E\phi \impl \calG,q\models^T\phi$ for any
formula~$\phi\in\SLBGf$.

%

\begin{example}
  Consider again the game of Fig~\ref{jhvgjjgjh} in
  Section~\ref{se::Definitions}. We have seen 
  that
  ${\calG,q_0 \models^{C}\phi}$ in Section~\ref{se::Definitions}, and
  that $\calG,q_0 \not\models^{E}\phi$ in the proof of
  Prop.~\ref{PropSLCGSideoupasSideRemoveDepende}.
  With timeline dependences, we have $\calG,q_0
  \models^{T}\phi$. Indeed, now $\theta(w)(z)(q_0\cdot q_2)$ may depend
  on~$w(x_A)(q_0)$ and~$w(x_B)(q_0)$: we~could then have e.g. 
  $\theta(w)(z)(q_0\cdot q_2)=p_1$ when~$w(x_A)(q_0)=q_2$,
  and $\theta(w)(z)(q_0\cdot q_2)=p_2$ when $w(x_A)(q_0)=q_1$.
  It~is easily checked that this map is a $T$-witness of~$\phi$ for~$q_0$.
%
\end{example}

\subparagraph{Comparison of $\models^E$ and~$\models^T$.}
As explained at the end of Section~\ref{sec-depend}, the proof of
Prop.~\ref{prop-negSL1G} actually shows the following result:
\begin{proposition}
\label{prop-ET-SL1G}
  For any game~$\calG$
  with initial state $q_0$, and any formula $\phi \in \SLOneGf$,
  it~holds
  $\calG,q_0 \models^{E} \phi  \iff
  \calG,q_0 \models^{T}  \phi$.
\end{proposition}
We~now prove
that this does not extend to~\SLCGf and~\SLDGf:

%

\begin{restatable}{proposition}{proptimelinesat}
\label{prop-timelinesat}
The relations $\models^E$ and $\models^T$ differ on~\SLCGf, 
as~well as on~\SLDGf.
\end{restatable}

\begin{proof}
  \begin{figure}[t]
    \begin{minipage}{.48\linewidth}
    \centering
    \begin{tikzpicture}[every node/.style={draw}]
      \node [carre5, draw=\FigColB, fill=\FigColB!40!white] (root) at (0,0) {$q_0$};
      \node [rond6, draw=\FigColA, fill=\FigColA!40!white] (sqr1) at (1,-.5) {$a$};
      \node [rond6, draw=\FigColA, fill=\FigColA!40!white] (sqr2) at (1,.5) {$b$};
      \node [draw=none] (A) at (2.2,-.5) {$\ p_1$};
      \node [draw=none] (B) at (2.2,.5) {$\ p_2$};
      \draw [-latex'] (root) to (sqr1); \draw [-latex'] (root) to (sqr2);
      \draw [-latex'] (sqr1) to (A); \draw [-latex'] (sqr2) to (B);
      \draw [-latex'] (sqr1) to (B); \draw [-latex'] (sqr2) to (A);
    \end{tikzpicture}
    \caption{$\models^E$ and $\models^T$ differ on~\SLCGf}
    \label{fig-ETCG}
    \end{minipage}\hfill
    \begin{minipage}{.48\linewidth}
      \centering
    \begin{tikzpicture}[every node/.style={draw}]
      \node [carre5, draw=\FigColB, fill=\FigColB!40!white] (root) at (0,0) {$q_0$};
      \node [rond6, draw=\FigColA, fill=\FigColA!40!white] (sqr) at (1,-.5) {$a$};
      \node [diamond, inner sep=0pt, minimum size=7mm, draw=\FigColC, fill=\FigColC!40!white] (diam) at (1,.5) {$b$};
      \node [draw=none] (A) at (2.2,-.5) {$\ p_1$};
      \node [draw=none] (B) at (2.2,.5) {$\ p_2$};
      \draw [-latex'] (root) to (sqr); \draw [-latex'] (root) to (diam);
      \draw [-latex'] (sqr) to (A); \draw [-latex'] (sqr) to (B);
      \draw [-latex'] (diam) to (B); \draw [-latex'] (diam) to (A);
    \end{tikzpicture}
    \caption{$\models^E$ and $\models^T$ differ on~\SLDGf}
    \label{fig-ETDG}
    \end{minipage}
  \end{figure}
  For \SLCGf, we~consider the game structure of Fig.~\ref{fig-ETCG},
  and formula 
  \begin{alignat*}1
    \phi_C &=\exists y.\ \forall x_A.\ \exists x_B .\
    \bigwedge
    \begin{cases}
      \assign{\CircleFill[\FigColA]{}\mapsto y;\ \BoxFill[\FigColB]{}\mapsto x_A}.\  \F p_1 \
      \\
      \assign{\CircleFill[\FigColA]{}\mapsto y;\ \BoxFill[\FigColB]{}\mapsto x_B}.\  \F p_2
    \end{cases}
  \end{alignat*}
  We~first notice that $\calG,q_0\not\models^E\phi_C$: indeed, in
  order to satisfy the first goal under any choice of~$x_A$, the
  strategy for~$y$ has to point to~$p_1$ from both~$a$ and~$b$. But
  then no choice of~$x_B$ will make the second goal true.

  On the other hand, considering the timeline semantics, strategy~$y$
  after~$q_0\cdot a$ and~$q_0\cdot b$ may depend on the choice
  of~$x_A$ in~$q_0$. When $w(x_A)(q_0)=a$, we~let
  $\theta(w)(y)(q_0\cdot a)= p_1$ and $\theta(w)(y)(q_0\cdot b)= p_2$
  and $\theta(w)(x_B)(q_0)= b$, which makes both goals hold
  true. Conversely, if $w(x_A)(q_0)=b$, then we~let
  $\theta(w)(y)(q_0\cdot b)= p_1$ and $\theta(w)(y)(q_0\cdot a)= p_2$
  and $\theta(w)(x_B)(q_0)= a$.

  For \SLDGf, we~consider the game of Fig.~\ref{fig-ETDG}, and easily
  prove that formula~$\phi_D$ below has a \ref{eq-T}-witness but no
  \ref{eq-E}-witness:
  \begin{xalignat*}1
      \phi_D &=\exists y.\ \forall x_A.\ \forall x_B.\ \forall z.\
    \bigvee
    \begin{cases} 
      \assign{\CircleFill[\FigColA]{}\mapsto  y ; \BoxFill[\FigColB]{}\mapsto  x_A ; \DiamondFill[\FigColC]{} \mapsto z}.\ \F p_1 
      \\
      \assign{\CircleFill[\FigColA]{}\mapsto  y ; \BoxFill[\FigColB]{}\mapsto  x_B ; \DiamondFill[\FigColC]{} \mapsto z}.\ \F p_2  
    \end{cases}
    \tag*{\qedsymbol}
  \end{xalignat*}
  \let\qed\relax
\end{proof}

\subparagraph{The syntactic vs.\ semantic negations.}
While both semantics differ, we~now prove that the situation
w.r.t. the syntactic vs.~semantic negations is similar.  First,
following Prop.~\ref{prop-ET-SL1G} and~\ref{prop-negSL1G}, the
two negations coincide on \SLOneGf under the
timeline semantics. Moreover:





\begin{restatable}{proposition}{ThmNegationsAndCoDeuxTroisAvrilOne}
  \label{PropNegationsAndCo::DeuxTroisAvril::ccr}
  \label{PropClotureNegationChapitreIntroDeDependencyProblem}
  For any formula $\phi$ in \SLBGf, for any game \calG and any
  state~$q_0$,
  we~have
  $\calG,q_0 \models^{T} \phi \impl \calG,q_0 \not\models^{T} \neg \phi$.
\end{restatable}

\begin{proof}[Sketch of proof]
  Write $\phi = (Q_i x_i)_{1\leq i \le l} \xi(\beta_j
  \varphi_j)_{j \le n}$.
  For a contradiction, assume that there
  exist two maps~$\theta$ and~$\bar\theta$ witnessing $\calG,q_0
  \models^{T} \phi$ and $\calG,q_0 \models^{T} \neg
  \phi$, respectively.  Then for any strategy valuations~$w$
  and~$\bar w$ for $\calV^{\forall}$ and $\calV^{\exists}$, we~have
  that $\calG,q_0\models_{\theta(w)} \xi(\beta_j\phi_j)_j$ and
  $\calG,q_0\models_{\bar \theta(\bar w)}
  \neg\xi(\beta_j\phi_j)_j$. We~can then inductively (on~histories and
  on the sequence of quantified variables) build a strategy
  valuation~$\chi$ on~$\calV$ such that
  $\theta(\chi_{|\calV^{\forall}}) =
  \bar\theta(\chi_{|\calV^{\exists}})=\chi$. Then under valuation~$\chi$,
  both $\xi(\beta_j\phi_j)_j$ and $\neg\xi(\beta_j\phi_j)_j$ hold
  in~$q_0$, which is impossible.
\end{proof}

\begin{restatable}{proposition}{ThmNegationsAndCoDeuxTroisAvrilDeux}
  \label{PropNegationsAndCo::DeuxTroisAvril::ccr:::deux}
  There exists
    a formula $\phi \in \SLBGf$,
    a (turn-based) game~$\calG$ and a state~$q_0$
    such that
    $\calG,q_0 \not\models^{T}\phi$
    and
    $\calG,q_0 \not\models^{T}\neg\phi$.
    
\end{restatable}

  \section{The fragment \protect\SLEGf}
  \label{se::SLEG}\label{sec-sleg}
In this section, we focus on the timeline semantics~$\models^{T}$.
%
We~exhibit a fragment \SLEGf of \SLBGf, containing \SLCGf and~\SLDGf,
for which the syntactic and semantic negations coincide, and for which
we prove model-checking is in \EXPTIME[2]:
\begin{restatable}{theorem}{MainTheoremeDuChapSLEGPartII}
  \label{ThmModelCheck}
  For any~$\phi\in\SLEGf$
  and any state~$q_0$, 
  it holds:
  \(
  \calG,q_0 \models^{T} \phi \iff \calG,q_0 \not\models^{T} \neg\phi\).
%
  Moreover, model checking \SLEGf for the timeline semantics is
  \EXPTIME[2]-complete. 
\end{restatable}

\subsection{Semi-stable sets.}
For~$n\in\bbN$, we~let $\ZO^n$ be the
set of mappings from~$[1,n]$ to~$\{0,1\}$.  We~write $\bfzero^n$
(or~$\bfzero$ if the size~$n$ is clear) for the function that maps all
integers in~$[1,n]$ to~$0$, and $\bfone^n$ (or~$\bfone$) for the
function that maps $[1,n]$ to~$1$.  The \emph{size} of $f\in\ZO^n$ is
defined as $\size f=\sum_{1\leq i\leq n} f(i)$.  For~two elements~$f$
and~$g$ of~$\ZO^n$, we~write $f\leq g$ whenever $f(i)=1$ implies
$g(i)=1$ for all~$i\in[1,n]$. Given $B^n\subseteq \ZO^n$, we~write
$\up{B^n}=\{g\in \ZO^n \mid \exists f\in B^n,\ f\leq g\}$.  A~set
$F^n\subseteq \ZO^n$ is \emph{upward-closed} if $F^n=\up{F^n}$.
Finally, for~$f,g\in\{0,1\}^n$, we~define:
\begin{alignat*}3
  \bar f&\colon i \mapsto 1-f(i)&
  \qquad
  f\inter g&\colon i\mapsto \min\{f(i),g(i)\} 
  \qquad&
  f\union g&\colon i\mapsto \max\{f(i),g(i)\}.
\end{alignat*}

We~then introduce the notion of \semist sets, on which the definition
of~\SLEGf relies:
  a~set $F^n\subseteq \{0,1\}^n$ is \emph{\semist} if 
  for any~$f$ and~$g$ in~$F^n$, it~holds that
  \[
  \forall s\in\{0,1\}^n. \qquad 
  (f\inter s) \union (g\inter \bar s) \in F^n \text{ or } 
  (g\inter s) \union (f\inter \bar s) \in F^n.
  \]

\begin{wrapfigure}[11]{r}{7cm}
\centering
\begin{tikzpicture}
\tikzstyle{color0}=[rouge];
\tikzstyle{color1}=[vert];
\path[use as bounding box] (-.5,.2) -- (5.5,-2.5);
\begin{scope}
\path(0,0) node (a) {};
\foreach \i in {1,0,1,1,1,0}
 {\draw (a) node[carre4,color\i] {$\i$};
  \path (a) -- +(4.1mm,0) node[coordinate] (a) {};}
\path (1.025,.4) node {$f$};
\draw[snake=brace] (1.42,-.3) -- (-.205,-.3) node[midway,below] (fs) {$\vphantom{h}s$};
\draw[snake=brace] (2.255,-.3) -- (1.45,-.3) node[midway,below] (fbs) {$\vphantom{h}\bar s$};
\end{scope}
\begin{scope}[xshift=3cm]
\path(0,0) node (a) {};
\foreach \i in {0,0,1,1,0,1}
 {\draw (a) node[carre4,color\i] (f\i) {$\i$};
  \path (a) -- +(4.1mm,0) node[coordinate] (a) {};}
\path (1.025,.4) node {$g$};
\draw[snake=brace] (1.42,-.3) -- (-.205,-.3) node[midway,below] (gs) {$\vphantom{h}s$};
\draw[snake=brace] (2.255,-.3) -- (1.45,-.3) node[midway,below] (gbs) {$\vphantom{h}\bar s$};
\end{scope}
\begin{scope}[yshift=-1.8cm]
\path(0,0) node (a) {};
\foreach \i in {0,0,1,1,1,0}
 {\draw (a) node[carre4,color\i] {$\i$};
  \path (a) -- +(4.1mm,0) node[coordinate] (a) {};}
\path (1.025,-.4) node {$(g\inter s)\union (f\inter \bar s)$};
\draw[snake=brace] (-0.205,.25) -- (1.42,.25) node[midway,coordinate,above=1.5mm] (c11) {};
\draw[snake=brace] (1.45,.25) -- (2.255,.25) node[midway,coordinate,above=1.5mm] (c12) {};
\end{scope}
\begin{scope}[yshift=-1.8cm,xshift=3cm]
\path(0,0) node (a) {};
\foreach \i in {1,0,1,1,0,1}
 {\draw (a) node[carre4,color\i] {$\i$};
  \path (a) -- +(4.1mm,0) node[coordinate] (a) {};}
\path (1.025,-.4) node {$(f\inter s)\union (g\inter \bar s)$};
\draw[snake=brace] (-0.205,.25) -- (1.42,.25) node[midway,coordinate,above=1.5mm] (c21) {};
\draw[snake=brace] (1.45,.25) -- (2.255,.25) node[midway,coordinate,above=1.5mm] (c22) {};
\end{scope}
\begin{scope}[dashed]
\draw[-latex'] (gs.-90) -- (c11);
\draw[-latex'] (fbs.-90) -- (c12);
\draw[-latex'] (fs.-90) -- (c21);
\draw[-latex'] (gbs.-90) -- (c22);
\end{scope}
\end{tikzpicture}
\caption{Semi-stability: if $f$ and $g$ are in~$F^n$, then one of $(g\inter s) \union (f\inter\bar s)$ and $(f\inter s)\union (g\inter\bar s)$ must be in~$F^n$.}
\label{fig-semist}
\end{wrapfigure}
Semi-stability is illustrated on Fig.~\ref{fig-semist}.

\begin{example}
  Obviously, the set $\ZO^n$ is \semist, as well as the empty set.
  For~$n=2$, the set
  $\{(0,1), (1,0)\}$ is easily seen not to be \semist: taking
  $f=(0,1)$ and $g=(1,0)$ with $s=(1,0)$, we~get
  $(f\inter s)\union(g\inter\bar s)=(0,0)$ and
  $(g\inter s)\union(f\inter\bar s)=(1,1)$.
  Similarly, $\{(0,0),(1,1)\}$~is not \semist. It~can be checked
  that any other subset of~$\ZO^2$ is \semist.
\end{example}

\medskip
We~then define 
\SLEGf~\footnote{We~name our fragment~\SLEGf
  as it comes as
  a natural continuation after fragments
  \SLAGf~\cite{mogavero2014behavioral},
  \SLBGf~\cite{Mogavero:2014:RSM:2656934.2631917}, and \SLCGf and
  \SLDGf~\cite{SLCGandDG2013Faux}.} as follows:
 \begin{xalignat*}2
  \SLEGf \ni \formule &\coloncolonequals
  \forall x.\formule \mid \exists x.\formule \mid \xi 
  &
  \xi &\coloncolonequals
  F^n((\beta_i)_{1\leq i\leq n})
  \\
  \beta &\coloncolonequals
  \assign{\sigma}.\ \psi
  & 
  \psi &\coloncolonequals
  \neg\psi \mid \psi\vee\psi \mid \X\psi \mid \psi\Until\psi \mid p
\end{xalignat*}
where $F^n$ ranges over \semist~subsets of~$\ZO^n$, for all~$n\in\bbN$.
%
%
The~semantics of the operator~$F^n$ is defined~as
\[
\calG,q \models_{\chi} F^n((\beta_i)_{i\leq n}) 
\quad \iff\quad
\text{letting~$f\in\ZO^n$ s.t. $f(i) =1\text{ iff } \calG,q \models_{\chi}\beta_i$, it~holds } 
f\in F^n.
\]

Notice that if $F^n$ would range over all subsets
of~$\ZO^n$, then this definition would exactly correspond to~\SLBGf.
Similarly, the case where $F^n=\{\bfone^n\}$
corresponds to~\SLCGf, while $F^n=\ZO^n\setminus \{\bfzero^n\}$ gives
rise  to~\SLDGf.

\begin{example}\label{ex-Nash2}
  Consider the following formula, expressing the existence of a Nash
  equilibrium for two players with respective \LTL objectives $\psi_1$
  and~$\psi_2$:
  \begin{xalignat*}1
    \exists x_1.\exists x_2. 
    \forall y_1.\forall y_2.&\ \bigwedge
    \left\{\begin{array}{@{}l@{}}
      (\assign{A_1\mapsto y_1 ; A_2\mapsto x_2}. \psi_1) \Rightarrow
      (\assign{A_1\mapsto x_1 ; A_2\mapsto x_2}. \psi_1)
      \\
      (\assign{A_1\mapsto x_1 ; A_2\mapsto y_2}. \psi_2) \Rightarrow
      (\assign{A_1\mapsto x_1; A_2\mapsto x_2}. \psi_2)
    \end{array}\right.
    \label{eq-nash}
  \end{xalignat*}
  This formula has four goals, and it corresponds to the set
  \[
  F^4 = \{(a,b,c,d) \in \ZO^4 \mid a \le b\ \text{and}\ c \le d\}
  \]
  Taking $f=(1,1,0,0)$ and $g=(0,0,1,1)$, with $s=(1,0,1,0)$
  we~have
  $(f\inter s)\union (g\inter\bar s) = (1,0,0,1)$ and
  $(g\inter s)\union (f\inter\bar s) = (0,1,1,0)$, none of which is in~$F^4$. 
  Hence our formula is not (syntactically) in~\SLEGf.
\end{example}

\begin{restatable}{proposition}{egag}
  \label{prop-egag}
\SLEGf contains \SLAGf. The inclusion is strict (syntactically).
\end{restatable}

\subsection{Properties of \semist\ sets}
\subparagraph{Transformation into an upward-closed set by bit flipping.}
 Fix a vector~$b\in \ZO^n$. We~define the operation $\flip_b\colon \ZO^n \to
 \ZO^n$ that maps any vector~$f$ to $(f\inter b)\union(\bar f \inter \bar b)$.
 In~other terms, $\flip_b$~flips the $i$-th bit of its argument if $b_i=0$, and
 keeps this bit unchanged if $b_i=1$. 
 In~\SLEGf, flipping bits amounts to negating the corresponding goals.
The~first part of the following lemma thus indicates that our
definition for \SLEGf is somewhat sound.%
%
%
\begin{restatable}{lemma}{dkfzefaaaaze}
  \label{FLipLemmaPropiSLEG::dzal}\label{lemma-upwardcl}
If~$F^n\subseteq \ZO^n$ is semi-stable, then so is $\flip_b(F^n)$.
Moreover, for any semi-stable set~$F^n$, there exists $B\in \ZO^n$
such that $\flip_B(F^n)$ is upward-closed (i.e. for any
$f\in\flip_B(F^n)$ and any $s\in\ZO^n$, we~have $f\union
s\in\flip_B(F^n)$).
\end{restatable}

\begin{remark}
  Notice that being upward-closed is not a sufficient condition for
  being semi-stable. For instance, the set~$F^n= \up{\{
    (0,0,1,1) ; (1,1,0,0) \}}$ is not semi-stable.
\end{remark}

%
%

\subparagraph{Defining quasi-orders from \semist sets.}
For $F^n\subseteq \ZO^n$,
we write $\overline{F^n}$ for the complement of~$F^n$. Fix such a set~$F^n$, and pick
$s\in\ZO^n$. For any
$h\in\ZO^{n}$, we define
\begin{xalignat*}1
  \mathbb{F}^n(h,s)&=\{ h' \in \ZO^n \mid (h\inter
  s)\union (h'\inter \overline{s}) \in F^n\}
  \\
  \overline{\mathbb{F}^n}(h,s)&=\{ h' \in \ZO^{n} \mid (h\inter
  s)\union (h'\inter \overline{s}) \in \overline{F^n}\}
\end{xalignat*}
Trivially $\mathbb{F}^n(h,s) \cap \overline{\mathbb{F}^n}(h,s)=\emptyset$ and
$\mathbb{F}^n(h,s) \cup \overline{\mathbb{F}^n}(h,s)=\ZO^n$.
If we assume $F^n$ to be semi-stable, then the family $(\mathbb{F}^n(h,s))_{h\in \ZO^n}$
enjoys the following property:
\begin{restatable}{lemma}{RTOSNE}
  \label{lemmaOrderingFsh}
  Fix a semi-stable set $F^n$ and $s\in\ZO^n$. For any $h_1, h_2\in\ZO^{n}$, either $\mathbb{F}^n(h_1,s) \subseteq \mathbb{F}^n(h_2,s)$ or $\mathbb{F}^n(h_2,s)\subseteq \mathbb{F}^n(h_1,s)$.
\end{restatable}

Given a semi-stable set $F^n$ and $s\in\ZO^n$, we can use the
inclusion relation of Lemma~\ref{lemmaOrderingFsh} to define a
relation $\preceq_s^{F^n}$ (written~$\preceq_s$ when~$F^n$ is clear)
over the elements of $\ZO^n$.
It~is defined as follows:
$h_1\preceq_s h_2$ if, and only~if,
$\mathbb{F}^n(h_1,s)\subseteq \mathbb{F}^n(h_2,s)$. 

This relation is a quasi-order:
its reflexiveness and
transitivity both follow from the reflexiveness and
transitivity of the inclusion relation~$\subseteq$.
By~Lemma~\ref{lemmaOrderingFsh}, this quasi-order is total.
Intuitively,
$\preceq_s$ orders the elements of $\ZO^{n}$ based on how ``easy'' it
is to complete their restriction to~$s$ so that the completion belongs
to~$F^n$.  In~particular, only the indices on which $s$ take value~$1$
are used to check whether $h_1\preceq_s h_2$: given
$h_1,h_2\in \ZO^n$ such that $(h_1\inter s)=(h_2\inter s)$, we~have
$\mathbb{F}(h_1,s)=\mathbb{F}(h_2,s)$, and $h_1=_s h_2$.
\leavevmode

\noindent
\begin{minipage}{9.5cm}
\begin{example}
  Consider for instance the \semist set~$F^3=\{(1,0,0),(1,1,0),
  (1,0,1), (0,1,1),(1,1,1)\}$ represented on the figure opposite, and
  which can be shown to be \semist. Fix~$s=(1,1,0)$. Then $\mathbb
  F^3((0,1,\star),s)= \ZO^2\times\{1\}$, while $\mathbb
  F^3((1,1,\star),s)=\mathbb F^3((1,0,\star),s)=\ZO^3$ and $\mathbb
  F^3((0,0,\star),s)=\emptyset$. It~follows that $(0,0,\star)
  \preceq_s (0,1,\star) \preceq_s (1,0,\star) =_s (1,1,\star)$.
\end{example}
\end{minipage}\hfill
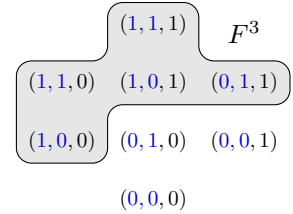
\begin{minipage}{4.2cm}
  \centering 
  \begin{tikzpicture}[yscale=1.1]
    %
    %
  \begin{scope}[xshift=6cm,yshift=-5mm,xscale=0.8, yscale=.88]
    \draw[thin, rounded corners=2mm,fill=black!10!white] (-1.5,0.5) to (-0.75,0.5) to (-0.75,1.6-0.3) to (2.25,1.6-0.3) to (2.25,1.6+0.3) to (0.75,1.6+0.3) to (0.75,2.4+0.3) to (-0.75,2.4+0.3) to (-0.75,1.6+0.3) to (-2.25,1.6+0.3) to (-2.25,0.5) -- (-1.5,0.5);
    \def\colbl{blue!80!black}
    \node [scale=0.8] at (0,0) { $(\color{\colbl}0,0\color{black},0)$};
    \node [scale=0.8] at (0,0.8) { $(\color{\colbl}0,1\color{black},0)$};
    \node [scale=0.8] at (1.5,0.8) { $(\color{\colbl}0,0\color{black},1)$};
    \node [scale=0.8, ] at (-1.5,0.8) { $(\color{\colbl}1,0\color{black},0)$};
    \node [ scale=0.8] at (-1.5,1.6) { $(\color{\colbl}1,1\color{black},0)$};
    \node [ scale=0.8] at (0,1.6) { $(\color{\colbl}1,0\color{black},1)$};
    \node [ scale=0.8]at (1.5,1.6) { $(\color{\colbl}0,1\color{black},1)$};
    \node [ scale=0.8] at (0,2.4) { $(\color{\colbl}1,1\color{black},1)$};
    \node at (1.5,2.30) {$F^3$};
    \end{scope}
    \end{tikzpicture}
\end{minipage}

%


\subsection{Sketch of proof of Theorem~\ref{ThmModelCheck}}

We~encode the \LTL formulas as parity automata, so that, by keeping
track of a vector of states of those automata, the goals to be
fulfilled are encoded as parity winning conditions. We~use
vectors~$s\in\ZO^n$ to represent the set of goals still being
``monitored'' after a finite history~$\rho$: for the
quasi-order~$\preceq_s$, there exist optimal elements~$b_{q,d,s}$ that
can be achieved from a given state~$q$ with a vector of states~$d$ of
the parity automata. There are two ways for the goals given
by~$b_{q,d,s}$ to be fulfilled: either by satisfying all those goals
along the same outcome, or by partitioning them along different
branches. This can be encoded in a two-player parity game (as in the
proof of Prop.~\ref{prop-negSL1G}); this has two consequences: first,
we~can effectively compute values~$b_{q,d,s}$, which provides the
\EXPTIME[2] algorithm (by checking whether $b_{q_0,d_0,\bfone}\in
F^n$); second, by determinacy of turn-based parity games, one
player has a winning strategy in each of those games. We~derive
timeline maps witnessing that the optimal elements~$b_{q,d,s}$ can be
achieved, and timeline maps witnessing that better elements cannot be
reached. These maps can be combined into global timeline maps
witnessing that $\calG,q_0\models^T \phi$ or $\calG,q_0\models^T
\non\phi$, depending whether $b_{q_0,d_0,\bfone}\in F^n$.

Finally, we prove that \SLEGf is, in a sense, maximal for
the first property of Theorem~\ref{ThmModelCheck}:%
\begin{restatable}{proposition}{MaximalitySLEGThmRestatableThingymarkOEOS}
  \label{ch:EG::PropMaximaliteSLEGvisavisSLBG}
  For any 
  non-\semist boolean set $F^n \subseteq
  \ZO^n$, there exists a \SLBGf formula \formule built on $F^n$, a game
  \calG and a state $q_0$ such that $\calG ,q_0 \not\models^{T} \neg
  \phi$ and $\calG ,q_0 \not\models^{T} \phi$.
\end{restatable}

Whether \SLEGf is also maximal for having a \EXPTIME[2] model-checking
algorithm remains open. Actually, we~do not know if \SLBGf model
checking is decidable under the timeline semantics. These questions
will be part of our future works on this topic.



  \clearpage
  \appendix
  \part*{Appendix}
\hfuzz=12pt\relax

\section{Proofs of Section~\ref{sec-depend}}
\label{app-sec3}
\subsection{Proof of Proposition~\ref{prop-negSL1G}}
\label{app-Prop4}
\PropositionquatreAvril*

\begin{proof}
%
We~begin with intuitive explanations before going into full details.
We~encode the satisfaction relation $\calG,q_0\models^E \phi$ into a
two-player turn-based parity game:
the~first player of the parity game will be in
charge of selecting the existentially-quantified strategies, and her
opponent will select the universally-quantified ones. This will be
encoded by replacing each state of~$\calG$ with a tree-shaped module
as depicted on Fig.~\ref{fig-quantgame}.
Following the strategy assignment of the \SLOneG formula~$\phi$,
the~strategies selected by those players will define a unique play,
along which the \LTL objective has to be fulfilled; this verification
is encoded into a (doubly-exponential) parity automaton.

We prove that $\calG,q_0\models^E \phi$ if, and only~if, the first
player~wins; conversely, $\calG,q_0\not\models^E \phi$ if the second
player wins. Both claims crucially rely on the existence of memoryless
optimal strategies for two-player parity games.
Finally, by~determinacy of those games, we~get the expected result.

\subparagraph{Building a turn-based parity game \calH from~$\calG$ and~$\phi$.}
  \label{subsectionSpecifyingcalHdansPreuveSLOneG}
  For the rest of the proof, we fix a game \calG and a \SLOneG formula
  $\phi= (Q_ix_i)_{i\leq l} \beta \varphi$.
  Each~state of~$\calG$ is replaced with a copy of the
  \emph{quantification game}
  depicted on Fig.~\ref{fig-quantgame}.
  A~quantification game~$\calQ_\phi$  is formally defined as follows:
  \begin{itemize}
  \item they involve two players $P_\exists$ and $P_\forall$;
  \item the set of states is $S_{\phi}=\{ \frakm\in \Act^* \mid
    0\leq \size{\frakm} \leq l\}$, thereby defining a tree of
    depth~$l+1$ with directions~$\Act$.  A~state~$\frakm$
    in~$S_{\phi}$ with $0\leq \size{\frakm}<l$ belongs to
    Player~$P_\exists$ if, and only if,
    $Q_{\size{\frakm}+1}=\exists$.
    States with $\size{\frakm}=l$ will have only one outgoing transition.
  \item from each~$\frakm$ with $0\leq \size{\frakm}<l$,
    for all~$a\in\Act$, there is a transition to $\frakm\cdot
    a$.  The~empty word $\epsilon\in S_{\phi}$ is the starting node of
    the quantification game, and has no incoming transitions; states
    with $\size{\frakm}=l$ (currently) have no outgoing
    transitions.
  \end{itemize}
  
  \begin{figure}[b]
    \centering
    \begin{tikzpicture}[scale=1]
      \node at (3.5,1.1){$\phi = \exists x_1.~ \forall x_2.~ \exists x_3.\ \beta\psi$};
      %
      \node [rond6] (root) at (0,1) {$\epsilon$};
      \node [carre5] (a) at (-1.6,0) {$a$};
      \draw [-latex'] (root) to (a);
      \node [carre5] (b) at (1.6,0) {$b$};
      \draw [-latex'] (root) to (b);
      \foreach \s/\a/\n in {a/a/-2.7,a/b/-1,b/a/1,b/b/2.7}
               {\node [draw, rond6] (\s\a) at (\n,-1) {$\vphantom{f}\s\a$};
                \draw [-latex'] (\s) to (\s\a);}
%
        \draw (2.4,.9) -- +(.7,0) node[midway,below=0mm,inner sep=0pt] (x1) {};
        \draw[dashed,opacity=.4,-latex'] (x1) edge[-latex',out=-150,in=0] (root);
        \draw (3.2,.9) -- +(.7,0) node[midway,below=0mm,inner sep=0pt] (x2) {};
        \draw[dashed,opacity=.4,-latex'] (x2) edge[-latex',out=-110,in=10] (b);
        \draw (4,.9) -- +(.7,0) node[midway,below=0mm,inner sep=0pt] (x3) {};
        \draw[dashed,opacity=.4,-latex'] (x3) edge[-latex',out=-90,in=40] (bb);
        \node (a) at (-4.5,-2) {};
        \foreach \s in {aa,ab,ba,bb}
          {\foreach \a in {a,b}
            {\path (a) -- +(1,0) node (a) {}
            node[draw,rounded corners=2mm] (\s\a) {$\vphantom{h}\s\a$};
            \draw[-latex'] (\s) -- (\s\a);
            \draw[dashed,-latex'] (\s\a.-90) -- +(-90:5mm);}}
      \draw[dotted,rounded corners=8mm] (0,-2.4) -- +(-4.8,0) -- +(0,4.3) -- +(4.8,0) -- +(0,0);
    \end{tikzpicture}
    \caption{The quantification game for $\phi = \exists x_1.~ \forall x_2.~ \exists x_3.\ \beta\psi$.}
    \label{fig-quantgame}
  \end{figure}
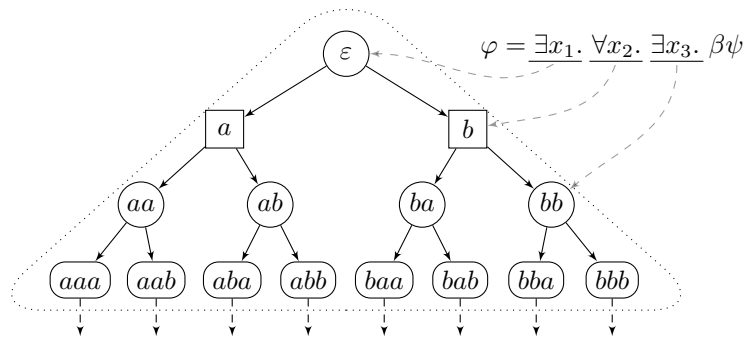
  
  A~leaf (i.e., a state~$\frakm$ with $\size{\frakm}=l$)
  in a quantification game represents a move vector of domain
  $\calV=\{x_i\mid 1\leq i\leq l\}$: 
  we~identify each~leaf~$\frakm$ with the move vector~$\frakm$,
  hence writing~$\frakm(x_i)$ for~$\frakm(i)$.
  
  \smallskip 
  
  We denote by $D$ a deterministic parity automaton over $2^\AP$
  associated with~$\varphi$.  We~write~$d_0$ for the initial
  state of~$D$.  Using quantification games, we can now define the turn-based
  parity game~\calH:
  \begin{itemize}
    \item it involves both players~$P_\exists$ and~$P_\forall$;
    \item  for each state~$q$ of~\calG and each state~$d$ of~$D$,
      \calH~contains a copy of the quantification game~$\calQ_\phi$,
      which we call the \emph{$(q,d)$-copy}. 
      Hence the set of states of~$\calH$ is the product of the state
      spaces of~$\calG$, $D$ and~$\calQ_\phi$.
    \item the transitions in~$\calH$ are of two types: 
    \begin{itemize} 
    \item internal transitions in each copy of the quantification game
      are preserved;
    \item consider a state $(q,d,\frakm)$ where $\size{\frakm}=l$;
      this is a leaf in the quantification game. If there exists a
      state~$q'$ such that $q'=\Delta(q,m_{\beta})$ where $m_{\beta}\colon
      \Agt \to \Act$ is the move vector over $\Agt$ defined by
      $m_{\beta}(A)=\frakm(i-1)$ where $x_i=\beta(A)$ (i.e.,
      assigning to each player~$A\in\Agt$ the action $\frakm(\beta(A))$),
      then we add a transition from~$(q,d,\frakm)$
      to~$(q',d',\epsilon)$ where $d'$ is the state of~$D$ reached
      from~$d$ when reading~$\labels(q')$. Notice that $(q,d,\frakm)$
      then has at most one outgoing transition.
    \end{itemize}
  \item the priorities are inherited from those in~$D$: state~$(q,d,\frakm)$
    has the same priority as~$d$.
  \end{itemize}

  \subparagraph{Correspondence between \calG and \calH.}
  We~define a correspondence between~$\calG$ and~$\calH$ through
  the notion of \emph{lanes}:
  \begin{definition}
    A \emph{lane} in \calG is a tuple $(\rho,u,b,t)$ made of 
    \begin{itemize}
    \item a history $\rho=(q_j)_{0\leq j\leq a}$ (for some integer~$a$);
    \item a function $u\colon \calV \times \Prefx(\rho) \to\Act$;
    \item an integer $b\in [0;l]$;
    \item a function $t\colon \{x_1, ..., x_b\} \to \Act$
      ($t$~is the empty function if~$b=0$);
    \end{itemize}
    and such that 
    \begin{equation}
    \forall 0\leq j<a .\quad
    \Delta(q_j,(m_j(\beta(A)))_{A\in\Agt})=q_{j+1}
    \qquad\text{ with}
    \begin{array}[t]{rcl}
      m_j\colon \calV &\to &\Act\\
      x &\mapsto &u(x,\rho_{\leq j})
    \end{array}
    \label{eq-coherence}
    \end{equation}
  \end{definition}
  
  We can then build a one-to-one application~$\HGp$ between histories 
  in~\calH and lanes in~\calG.
  With a history~$\pi$ in~\calH, written
  \[
  \pi = 
  \Bigl(\prod_{0\leq j< a}\  \prod_{0\leq i\leq l}\ (q_{j},d_{j},\frakm_{j,i})\Bigr)
  \cdot \prod_{0\leq i\leq b} (q_{a},d_{a},\frakm_{a,i}),
  \]
  having length $a\cdot (l+1)+b+1$ with $0\leq b <l$,
  we associate a lane $\HGp(\pi)=((q_j)_{j\leq a}, u, b, t)$
  with
    \begin{alignat*}2
     u\colon \calV\times \Prefx(\rho) \to{} &\Act    
     &\qquad
     t\colon \{x_1,...,x_b\} \to{} & \Act
     \\
     x_i,(q_j)_{j\leq c}\mapsto{} & \frakm_{c,i} \qquad (\forall c<a)
     &
     x_i \mapsto{} & \frakm_{a,i}
    \end{alignat*}
  
    The resulting function~$\HGp$ is clearly injective
    (different histories will correspond to different lanes), 
  but also surjective.
  To~prove the latter statement, we~build the inverse function~$\GHp$:
  for a lane $((q_j)_{j\leq a}, u, b, t)$, we set
  $\GHp((q_j)_{j\leq a}, u, b, t) =\pi$ where $\pi$ is the history
  in~\calH of length $a\cdot (l+1)+ b+1$ defined as
  \[
  \pi=
  \prod_{0\leq j< a}\
  \prod_{0\leq i\leq l} 
  \bigl(
  q_{j},d_j,
  u(x_i,(q_{j'})_{j'\leq j})
  \bigr)
  \cdot\prod_{0\leq i\leq b}
  \bigl(
  q_{a},d_{a},
  t(x_i,(q_j)_{j\leq a})
  \bigr)
  \]
  where $d_j$ is the state of~$D$ reached on input $(q_k)_{0\leq k\leq j-1}$. 

  
  Because of the coherence condition~\eqref{eq-coherence},
  $\GHp((q_j)_{j\leq a}, u, i, t)$ is indeed a history in~$\calH$.
  From the definitions, one can easily check that
  \[ 
  \GHp(\HGp(\pi))=\pi
  \] 
  and deduce that $\GHp$ is the inverse function of~$\HGp$; therefore
  \begin{lemma}
    The application $\HGp$ is a bijection between lanes of~\calG and
    histories in~\calH, and $\GHp$ is its inverse function.
  \end{lemma}
  
  \subparagraph{Extending the correspondence.}
  We can use $\HGp$ to describe another correspondence~$\HG$ between
  (positional) strategies for~$P_\exists$ in~\calH and (elementary) maps
  in~\calG.
  Remember that a map is a function $\theta\colon (\Hist_\calG
  \to \Act)^{\calV^\forall} \to (\Hist_\calG \to \Act)^{\calV}$.
  Remember also that if
  $Q_j=\forall$, then $\theta(w)(x_i)(\rho)=w(x_i)(\rho)$, so that we
  only have to define the map for the existentially-quantified
  variables.
  
  Formally, the application $\HG$ takes as input a 
  strategy~$\delta$ for player~$P_\exists$ in~\calH, and returns a map
  in~$\calG$. It~will enjoy the following properties:
  \begin{itemize}
  \item for any finite outcome~$\pi$ of~$\delta$ in~$\calH$
    ending at the root of a quantification
    game, there exists a function $w$ such that
    $\HGp(\pi)=(\rho,u,0,t_\emptyset)$ where $\rho$~is the outcome of
    $\HG(\delta)(w)$ in~$\calG$ under the assignment defined
    by~$\beta$;
  \item conversely, for any path~$\rho$ in~$\calG$ 
    that
    is an outcome of~$\HG(\delta)(w)$ for some~$w$ and under the
    assignment defined by~$\beta$, then letting
    $u(x,\rho')=\HG(\delta)(w)(x)(\rho')$, we~have that
    $(\rho,u,0,t_\emptyset)$ is a lane in~$\calG$ and
    $\GHp(\rho,u,0,t_\emptyset)$ is an outcome of~$\delta$ in~$\calH$
    ending in the root of a quantification game.
  \end{itemize}

  We~fix~$\delta$, and for all $w$, $\rho$ and~$x_i$,   
  we~define $\HG(\delta)(w)(x_i)(\rho)$ by a double induction, 
  first on the length of the history~$\rho$ in~$\calG$,
  and second on the sequence of variables~$x_i$.
  \begin{itemize}
  \item \textbf{initial step:} we~begin with the case where $\rho$ is
    the single state~$q_0$.
      %
%
    We~proceed by induction on existentially-quantified variables,
    merging the initialization step with the induction step as they
    are similar.  Consider an existentially-quantified variable~$x_i$
    in~$\calV$.  Given $w\colon \calV^\forall \times \Prefx(\rho) \cup
    \{\rho\} \to \Act$, we~define a function $t_{i,w}\colon
    [x_1;x_{i-1}] \to \Act$ such that $t_{i,w}(x)=w(x,q_0)$ for $x\in
    \calV^\forall\cap[x_1;x_{i-1}]$, and
    $t_{i,w}(x)=\HG(\delta)(w)(x)(q_0)$ for
    $x\in\calV^\exists\cap[x_1;x_{i-1}]$, assuming that they have been
    defined in the previous induction steps on variables.
      We can then create the lane $\lane_{i,w}=(\epsilon,u_\emptyset,i-1,t)$
      and define 
      \[ 
      \HG(\delta)(w)(x_i)(q_0)=\delta(\GHp(\lane_{i,w}))
      \]
      %

      Pick an outcome~$\pi$ of~$\delta$ in~$\calH$ of length~$l+2$,
      and write $\frakm$ for its $l+1$-st state: it~defines a
      valuation for the variables in~$\calV$, hence defining a move
      vector~$m_\beta$ under the assignment~$\beta$.
      in~$\Act$. By~construction of~$\calH$, this outcome ends in the
      state $(q_1,d_1,\epsilon)$ where $q_1=\Delta(q_0,m_\beta)$ and
      $d_1$ is the successor of the initial state~$d_0$ of~$D$ when
      reading $\labels(q_1)$. We~now prove that $q_0\cdot q_1$ is the
      outcome of $\HG(\delta)(w)$ for some~$w$.  For this,
      we~let~$w(x_i)=\frakm_i$ for all~$x_i\in\calV^{\forall}$.
      By~construction, $\HG(\delta)(w)(x_j)(q_0)$ precisely
      corresponds to $\frakm(j)$, for all
      $x_j\in\calV^{\exists}$. In~the end, under assignment~$\beta$,
      $\HG(\delta)(w)$ precisely returns the move vector~$m_\beta$,
      hence proving our result.

      The proof of the converse statement follows similar arguments:
      consider an outcome $\rho=q_0\cdot q_1$ of $\HG(\delta)(w)$ for
      some~$w$. The~lane $(\rho,u,0,t_\emptyset)$ defined with 
      $u(x,q_0)=\HG(\delta)(w)(x)(q_0)$ then
      corresponds through $\GHp$ to a play
      ending in $(q_1,d_1,\epsilon)$, 
      and visiting the leaf~$\frakm$ defined as $\frakm_i=u(x_i,q_0)$.
      By~construction, this is an outcome of~$\delta$ in~$\calH$.

    \item \textbf{induction step:} we consider a history~$\rho$ in~$\calG$,
      assuming we have define $\HG(\delta)(w)(x_i)(\rho')$ for all
      prefixes~$\rho'$ of~$\rho$, for all~$w$ and all
      variables~$x_i$.  We~now define
      $\HG(\delta)(w)(x_i)(\rho)$, by induction on the list of
      variables.  Again, the initialization
      step is merged with the induction step as they rely on the same
      arguments.
      %

      Consider an existentially-quantified variable~$x_i$, and
      $w\colon \calV^\forall \times \Prefx(\rho) \cup \{\rho\} \to
      \Act$.  We~define a function $t_{i,w}\colon [x_1;x_{i-1}] \to
      \Act$ where $t_{i,w}$ associate with $x\in\calV^\forall\cap
      [x_1;x_{i-1}]$ the action $w(x)(\pi)$, and with
      $x\in\calV^\exists\cap [x_1;x_{i-1}]$ the action
      $\HG(\delta)(w)(x)(\rho)$.  We~also define $u_w\colon
      \calV\times\Prefx(\rho) \to \Act$ as
      $u_w(x,\rho')=\HG(\delta)(w)(x)(\rho')$, for all
      prefixes~$\rho'$ of~$\rho$.
      We can then create the lane
      $\lane_{i,w}=(\pi,u_w,i-1,t_{i,w})$ and finally define
      \[ 
      \HG(\delta)(w)(x_i)(\rho)=\delta(\GHp(\lane_{i,w})).
      \]
      Using the same arguments as in the initial step, we prove our
      correspondence between the outcomes of~$\delta$ in~$\calH$ and
      the outcomes of~$\HG(\delta)$ in~$\calG$.
  \end{itemize}

  Notice that in the construction above, $\HG(\delta)(w)(x_i)(\rho)$
  may depend on the value of~$w(x_j,\rho')$ for $j>i$ and
  $\rho'\in\Prefx(\rho)$: indeed, in the inductive definition,
  we~define $\HG(\delta)(w)(x_j)(\rho')$ before defining
  $\HG(\delta)(w)(x_i)(\rho)$. Hence in general $\HG(\delta)$ is
  \emph{not} en elementary map.
  
  However, in case $\delta$ is memoryless, we~notice that
  $\HG(\delta)(w)(x_i)(\rho)$ only depends on value of~$\delta$ in the
  last state of the lane $\lane_{i,w}$, hence in particular not
  on~$u_w$. This removes the above dependence, and makes $\HG(\delta)$
  elementary.

  Finally, notice that we~can define a dual correspondence
  $\overline{\HG}$ relating strategies of Player~$P_\forall$ and
  elementary maps in~$\calG$ where existential and universal variables
  are swapped. 

  \subparagraph{Concluding the proof.}
  Using $\HG$, we~prove our final correspondence between~$\calH$ and~$\calG$:
  \begin{restatable}{lemma}{PropositionFromcalHtoCalGDansSLOneGProofConcludingTheProof}
    \label{PropositionWinningcalHdeltaThenWitnesssDansCalGChaPintroDepProv}
    Assume that $P_\exists$ is winning in~\calH and let $\delta$ be a
    positional winning strategy. Then the elementary map~$\HG(\delta)$
    is a witness that $\calG,q_0 \models^{E} \phi$.
    
    Similarly, assume that $P_\forall$ is winning in \calH and let
    $\overline{\delta}$ be a positional winning strategy.  Then the
    elementary map $\overline{\HG}(\overline{\delta})$ is a witness that
    $\calG,q_0 \models^{E} \neg\phi$.
  \end{restatable}
  \begin{proof}
    We prove the first point, the second one following similar
    arguments.  Assume that $P_\exists$ is winning in~\calH, and pick
    a memoryless winning strategy~$\delta$.  Toward a contradiction,
    assume further that $\HG(\delta)$ is not a witness of $\calG,q_0
    \models^E \phi$. Then
    there exists $w_0\colon \calV^\forall \to (\Hist_\calG \to \Act)$ s.t.
    \(
    \calG, q_0 \not\models_{\HG(\delta)(w_0)} \beta\phi
    \).
    We use $w_0$ to build a strategy $\overline{\delta}$ for
    Player~$P_\forall$ in~\calH.
    Given a history 
    \[ 
    \pi =
    \prod_{0\leq j< a}\ \prod_{0\leq i\leq l} (q_{j},d_{j},\frakm_{j,i})
    \cdot\prod_{0\leq i\leq b} (q_{a},d_{a},\frakm_{a,i})
    \]
    in~$\calH$,
    we define $\rho=\prod_{0\leq j\leq  a} q_{j}$ and set 
    $  \overline{\delta}(\pi) = \HG(\delta)(w)(x_b)(\eta) $
    where
    \begin{itemize}
    \item $w\colon \Prefx(\rho) \cup \{\rho\} \times (\calV^\forall
      \cap [x_1;x_b]) \to \Act$ is such that $w(\rho',x_i)$ is the
      action to be played for going from $\pi_{\leq \size{\rho'}\cdot
        (l+1)+i-1}$ to $\pi_{\leq \size{\rho'}\cdot (l+1)+i}$
      in~$\calH$;
    \item $\eta = \prod_{0\leq j< a}\ \prod_{0\leq i\leq l}
        (q_{j},d_{j},\frakm_{j,i}))$.
    \end{itemize}

    Write $\nu=(q_j)_{j\in\bbN}$ for the outcome of $\theta(w_0)$
    under strategy assignment~$\beta$ in~\calG. Then, by construction
    of $\overline{\delta}$, the~outcome of~$\delta$
    and~$\overline{\delta}$ in~\calH will visit the
    $(q_j,d_j)_{j\in\bbN}$-copies of the quantification game, where
    $d_j$ is the state reached by reading $(q_{j'})_{j'\leq j}$ in the
    deterministic automaton~$D$.
    Now,
    since $\calG, q_0 \not\models_{\HG(\delta)(w_0)} \beta\phi$,
    we get that $\nu$ does not satisfy $\phi$ and therefore the
    outcome of $\delta$ and $\overline{\delta}$ in~$\calH$ does not satisfy the
    parity condition. This is in contradiction with $\delta$ being the
    winning strategy of $P_\exists$, and proves that $\HG(\delta)$ must
    be a witness that $\calG,q_0 \models^E\phi$.
  \end{proof}
  
    Proposition~\ref{PropositionWinningcalHdeltaThenWitnesssDansCalGChaPintroDepProv},
    together with the determinacy of parity games~\cite{EJ91,Mos91}
    immediately imply that at least one of $\phi$ and~$\neg\phi$ must
  hold in~$\calG$ for~$\models^{E}$. This concludes our proof.
  %
\end{proof}

\subsection{Proof of Proposition~\ref{prop-CEdiff}}
\label{app-prop6}
\PropSLCGcodeAIAHIFHEFE*

\begin{proof}
The proof of
  Prop.~\ref{Chapitre3Framework::Prop::NeitherHold::31MArs2017}
  provides an example of a game and a formula in~\SLCGf where
  $\models^C$ and $\models^E$ differ. We~prove the result for~\SLDGf:
  consider again the game of
  Fig.~\ref{jhvgjjgjh} in Section~\ref{se::Definitions}.
%
  We~already proved that $\calG,q_0 \models^{C}\phi$; we~show that
  $\calG,q_0\not\models^E \phi$. For this, consider the following four
  valuations for the universally-quantified strategies:
  \begin{xalignat*}3
    w_1(x_A)(q_0) &= q_1 & w_1(x_B)(q_0) &= q_2 & w_1(y)(q_0\cdot q_1) & = p_2 \\
    w_2(x_A)(q_0) &= q_2 & w_2(x_B)(q_0) &= q_1 & w_2(y)(q_0\cdot q_1) & = p_1
  \end{xalignat*}
  (assuming that they coincide in any other situation).
  Let $\theta$ be an elementary $\phi$-map: then it must be such that
  $\theta(w_1)(z)(q_0\cdot q_2)=\theta(w_2)(z)(q_0\cdot q_2)$. Then:
  \begin{itemize}
  \item if $\theta(w_1)(z)(q_0\cdot q_2)=\theta(w_2)(z)(q_0\cdot
    q_2)=p_1$, then the first goal goes to~$p_2$ via~$q_1$, and the
    second goal goes to~$p_1$ via~$q_2$. None of those goals is fulfilled;
  \item if $\theta(w_1)(z)(q_0\cdot q_2)=\theta(w_2)(z)(q_0\cdot
    q_2)=p_2$, then the first goal goes to~$p_2$ via~$q_2$, and the
    second goal goes to~$p_1$ via~$q_1$. Again, both goals are missed.\popQED
  \end{itemize}
%
\end{proof}

\section{Proofs of Section~\ref{sec-timeline}}
\label{app-sec4}

\subsection{Proof of Proposition~\ref{PropNegationsAndCo::DeuxTroisAvril::ccr}}
\label{app-prop8}
\ThmNegationsAndCoDeuxTroisAvrilOne*
\label{app-Prop8}

\begin{proof}
  
  For a contradiction, assume that there exist two maps~$\theta$
  and~$\bar\theta$ witnessing $\calG,q_0 \models^{T} \phi$ and
  $\calG,q_0 \models^{T} \neg \phi$ resp.  Then
  \begin{alignat}1
    \forall w\colon \calV^\forall \to (\Hist \to \Act).\quad
     &\calG,q_0\models_{\theta(w)}
    \xi(\beta_j\varphi_j)_{j\leq n} \label{FormulaElemnegationpreuveun}
    \\
    \forall \overline{w}\colon \calV^\exists \to (\Hist \to \Act).\quad
    &\calG,q_0\models_{\overline{\theta}(\overline{w})} 
    \neg\xi(\beta_j\varphi_j)_{j\leq n} \label{FormulaElemnegationpreuvedeux}
  \end{alignat}

  From~$\theta$ and $\bar\theta$, we build a strategy valuation~$\chi$
  on~$\calV$ such that $\theta(\chi_{|\calV^{\forall}})
  = \bar\theta(\chi_{|\calV^{\exists}})=\chi$.
  By~Equations~\eqref{FormulaElemnegationpreuveun}
  and~\eqref{FormulaElemnegationpreuvedeux}, we~get that 
  $\calG,q_0\models_{\chi} \xi(\beta_j\varphi_j)_{j\leq n}$ and 
  $\calG,q_0\models_{\chi}\neg\xi(\beta_j\varphi_j)_{j\leq n}$,
  which for \LTL formulas is impossible.

  We~define~$\chi(x)(\rho)$ inductively on histories and on the list
  of quantified variables. When $\rho$ is the empty history~$q_0$, we
  consider two cases:
  \begin{itemize}
  \item if~$x_1\in\calV^\forall$, then $\bar\theta(\bar w)(x_1)(q_0)$ does
    not depend on~$\bar w$ at all, since $\bar\theta$ is a
    timeline-map. Hence we~let
    $\chi(x_1)(q_0)=\bar\theta(\bar w)(x_1)(q_0)$, for any~$\bar w$.
  \item similarly, if $x_1\in\calV^\exists$, we~let
    $\chi(x_1)(q_0)=\theta(w)(x_1)(q_0)$, which again does not depend
    on~$w$.
  \end{itemize}
  Similarly, when $\chi(x)(q_0)$ has been defined for
  all~$x\in\{x_1,...,x_{i-1}\}$, we~again consider two cases:
  \begin{itemize}
  \item if~$x_i\in\calV^\forall$, we~define $\bar
    w(x_j)(q_0)=\chi(x_j)(q_0)$ for
    all~$x_j\in\calV^\exists\cap\{x_1,...,x_{i-1}\}$, and let
    $\chi(x_i)(q_0)=\bar\theta(\bar w)(x_i)(q_0)$, which again does
    not depend on the value of~$\bar w$ besides those defined above;
  \item symmetrically, if~$x_i\in\calV^\exists$, we~define
    $w(x_j)(q_0)=\chi(x_j)(q_0)$ for
    all~$x_j\in\calV^\forall\cap\{x_1,...,x_{i-1}\}$, and let
    $\chi(x_i)(q_0)=\theta(w)(x_i)(q_0)$.
  \end{itemize}
  Notice that this indeed enforces that
  $\theta(\chi_{|\calV^\forall})(x_i)(q_0)=\chi(x_i)(q_0)$ when~$x_i\in\calV^\exists$,
  and 
  $\bar\theta(\chi_{|\calV^\exists})(x_i)(q_0)=\chi(x_i)(q_0)$ when~$x_i\in\calV^\forall$.
  
  \smallskip
  The induction step is proven similarly: consider a history~$\rho$
  and a variable~$x_i$, assuming that $\chi$ has been defined for all
  variables on all prefixes of~$\rho$, and for variables
  in~$\{x_1,...,x_{i-1}\}$ on~$\rho$~itself. Then:
  \begin{itemize}
  \item if~$x_i\in\calV^\forall$, we~define
    $\bar w(x_j)(\rho')=\chi(x_j)(\rho')$ for
    all~$x_j\in\calV$ and all~$\rho'\in\Prefx(\rho)$, and
    $\bar w(x_j)(\rho)=\chi(x_j)(\rho)$ for
    all~$x_j\in\calV^\exists\cap\{x_1,...,x_{i-1}\}$.
    We~then let
    $\chi(x_i)(\rho)=\bar\theta(\bar w)(x_i)(q_0)$, which does
    not depend on the value of~$\bar w$ besides those defined above;
  \item the construction for the case when $x_i\in\calV^\exists$ is
    similar.
  \end{itemize}
  As in the initial step, it~is easy to check that this construction enforces 
    $\theta(\chi_{|\calV^{\forall}})
  = \bar\theta(\chi_{|\calV^{\exists}})=\chi$, as required.
\end{proof}

\subsection{Proof of Proposition~\ref{PropNegationsAndCo::DeuxTroisAvril::ccr:::deux}}
\label{app-prop9}
\ThmNegationsAndCoDeuxTroisAvrilDeux*

\begin{proof}
  Consider the turn-based
  game~$\calG$ and the \SLBGf formula~$\phi$ of
  Fig.~\ref{FigureProblemWideNegationGame}.
\begin{figure}[t]
  \centering 
  \begin{tikzpicture}[scale=1]
    \node [carre5, draw=\FigColB, fill=\FigColB!40!white] (a) at (0,0) {$q_0$};
    \node [rond6, draw=\FigColA, fill=\FigColA!40!white] (b) at (-1,-0.7) {$q_1$}; 
    \node [rond6, draw=\FigColA, fill=\FigColA!40!white] (c) at (1,-0.7) {$q_2$};
    \node (p1) at (-1.5,-1.5) {$p_1$}; \node (p2) at (-0.5,-1.5) {$p_2$}; 
    \node (p3) at (0.5,-1.5) {$p_3$}; \node (p4) at (1.5,-1.5) {$p_4$};
    \draw [-latex'] (a) to (b);\draw [-latex'] (a) to (c);
    \draw [-latex'] (b) to (p1);\draw [-latex'] (b) to (p2);\draw [-latex', bend left=20] (b) to (p3);
    \draw [-latex'] (c) to (p4);\draw [-latex', bend right=20] (c) to (p2);\draw [-latex'] (c) to (p3);
\path (6.9,-.5) node{%
  \begin{minipage}{.8\linewidth}
  \begin{alignat*}1
    \phi =
    \forall x_1.\exists y_1.\exists y_2.\exists x_2.\ \bigwedge
    \begin{cases}
      \assign{\BoxFill[\FigColB]{}\mapsto y_1 ; \CircleFill[\FigColA]{}\mapsto x_1} \F p_2\\\qquad\Rightarrow 
      \assign{\BoxFill[\FigColB]{}\mapsto y_2 ; \CircleFill[\FigColA]{}\mapsto x_2} \F p_1
      \\[2mm]
      \assign{\BoxFill[\FigColB]{}\mapsto y_1 ; \CircleFill[\FigColA]{}\mapsto x_1} \F p_3\\\qquad\Rightarrow 
      \assign{\BoxFill[\FigColB]{}\mapsto y_2 ; \CircleFill[\FigColA]{}\mapsto x_2} \F p_4 
    \end{cases}
  \end{alignat*}
  \end{minipage}};
  \end{tikzpicture}
  \caption{A game \protect\calG and a formula~$\phi$ such that $\calG,q_0 \models^{T}\phi$ and
  $\calG,q_0 \models^{T}\neg\phi$}
  \label{FigureProblemWideNegationGame}
\end{figure}
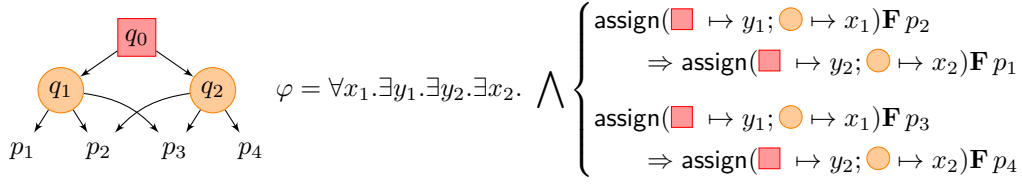
First, $\calG,q_0 \models^{T}\phi$, since for any choice of~$x_1$
and~$y_1$, one~of the goals holds vacuously, and the other one can be made true
by correctly selecting~$y_2$ and~$x_2$.
We~now prove that $\calG,q_0 \models^{T}\neg\phi$: since timeline
dependences are allowed, $\theta(w)(x_1)(q_0\cdot q_1)$ and
$\theta(w)(x_1)(q_0\cdot q_2)$ may depend on the values
of~$w(y_1)(q_0)$ and~$w(y_2)(q_0)$.  We~thus consider four cases:
\begin{itemize}
\item if $w(y_1)(q_0)=w(y_2)(q_0)=q_1$, then we~let
  $\theta(w)(x_1)(q_0\cdot q_1)=p_3$; then the second goal
  of~$\phi$ is not fulfilled, whatever~$w(x_2)$;
\item if $w(y_1)(q_0)=w(y_2)(q_0)=q_2$, then symmetrically, we~let
  $\theta(w)(x_1)(q_0\cdot q_2)=p_2$, so that the first goal of~$\phi$
  fails to hold for any~$w(x_2)$;
\item if $w(y_1)(q_0)=q_1$ and $w(y_2)(q_0)=q_2$, then we~let $\theta(w)(x_1)(q_0\cdot q_1)=p_2$, and again the first goal holds, whatever~$w(x_2)$;
\item if $w(y_1)(q_0)=q_2$ and $w(y_2)(q_0)=q_1$, then we~let $\theta(w)(x_1)(q_0\cdot q_2)=p_3$, and again the second goal fails to hold independently
  of~$w(x_2)$.
  \popQED
\end{itemize}
  %
\end{proof}

\section{Proofs of Section~\ref{sec-sleg}}

\subsection{Proof of Proposition~\ref{prop-egag}}
\egag*

\begin{proof}
Remember that boolean combinations in~\SLAGf follow the grammar
\(
\xi\coloncolonequals \xi\ou \beta \mid \xi\et\beta \mid\beta
\).
In terms of subsets of~$\ZO^n$, it~corresponds to considering sets
defined in one of the following two forms:
\begin{xalignat*}1
F^n_\xi & =\{f\in\ZO^n\mid f(n)=1\}\cup\{g\in\ZO^n\mid
g_{|[1;n-1]}\in F^{n-1}_{\xi'}\}\\
F^n_\xi&=\{f\in\ZO^n \mid
f(n)=1 \text{ and } f_{|[1;n-1]}\in F^{n-1}_{\xi'}\}
\end{xalignat*}
depending
whether $\xi(p_j)_j=\xi'(p_j)_j\vee p_n$ or
$\xi(p_j)_j=\xi'(p_j)_j\wedge p_n$. Assuming (by~induction) that
$F^{n-1}_{\xi'}$ is \semist, then we can prove that $F^n_\xi$ also~is.
We~detail the proof for the second case, the first case being similar.

Consider the case where $F^n_\xi=\{f\in\ZO^n \mid f(n)=1 \text{ and }
f_{|[1;n-1]}\in F^{n-1}_{\xi'}\}$.  Pick any two elements~$f$ and~$g$
in~$F^n_\xi$, and $s\in\{0,1\}^n$. Since $f(n)=g(n)=1$, we~have
$[(f\inter s)\union(g\inter\bar s)](n)= [(f\inter \bar
  s)\union(g\inter s)](n)=1$. Moreover, the restriction of $[(f\inter
  s)\union(g\inter\bar s)]$ and of $[(f\inter \bar s)\union(g\inter
  s)]$ to their first $n-1$ bits is computed from the restriction of
$f$, $g$ and~$s$ to their first $n-1$ bits. Since $F^{n-1}_{\xi'}$ is
\semist, one of $[(f\inter s)\union(g\inter\bar s)]_{[1;n-1]}$ and
$[(f\inter \bar s)\union(g\inter s)]_{[1;n-1]}$ belongs
to~$F^{n-1}_{\xi'}$, so that one of $[(f\inter s)\union(g\inter\bar
  s)]$ and $[(f\inter \bar s)\union(g\inter s)]$ is in~$F^n_{\xi}$.

\smallskip
That the inclusion is strict is proven by considering the \semist set
$H^3=\{\tuple{1,1,1},\penalty0\tuple{1,1,0},\penalty0
\tuple{1,0,1},\penalty0\relax \tuple{0,1,1}\}$.  Assume that it
corresponds to a formula in~\SLAGf: then the boolean combination
$\xi(x_1,x_2,x_3)$ of that formula must be in one of the following
forms:
\begin{xalignat*}4
\xi'(x_1,x_2) \et x_3 && 
\xi'(x_1,x_2) \ou x_3 && 
\xi'(x_1,x_2) \et \non x_3 && 
\xi'(x_1,x_2) \ou\non x_3.
\end{xalignat*}
It remains to prove that none of these cases corresponds to~$H^3$: the
first case does not allow $\tuple{1,1,0}$; the second case allows
$\tuple{0,0,1}$; the third case does not allow $\tuple{1,0,1}$; the
last case allows $\tuple{0,0,0}$.
\end{proof}

\subsection{Proof of Lemma~\ref{FLipLemmaPropiSLEG::dzal}}
The proof of Lemma~\ref{FLipLemmaPropiSLEG::dzal} will make use of the
following intermediary results. The first lemma shows that \SLEGf is
closed under (syntactic) negation.

\begin{lemma}\label{lemma-compl}
  $F^n$ is \semist if, and only~if, its complement~is.
\end{lemma}

\begin{proof}
Assume $F^n$ is not \semist, and pick $f$ and~$g$ in~$F^n$ and
$s\in\ZO^n$ such that none of $\alpha=(f\inter s)\union (g\inter \bar
s)$ and $\gamma=(g\inter s)\union (f\inter \bar s)$ are
in~$F^n$. It~cannot be the case that $g=f$, as this would imply
$\alpha=f\in F^n$. Hence $\alpha\not=\gamma$. We~claim that $\alpha$
and $\gamma$ are our witnesses for showing that the complement
of~$F^n$ is not \semist: both of them belong to the complement
of~$F^n$, and $(\alpha\inter s)\union(\gamma\inter \bar s)$ can be
seen to equal~$f$, hence it is not in the complement
of~$F^n$. Similarly for $(\gamma\inter s)\union(\alpha\inter\bar
s)=g$.
\end{proof}

\begin{lemma}\label{lemma-H}
  If~$F^n\subseteq \ZO^n$ is semi-stable, then for any $s\in\ZO^n$ and
  any non-empty subset~$H^n$ of $F^n$, it~holds that
  \[
  \exists f\in H^n.\ 
  \forall g\in H^n.\
  (f\inter s) \union (g\inter \bar s) \in F^n.
  \]
\end{lemma}

\begin{proof}
  For a contradiction, assume that there exist $s\in\ZO^n$ and $H^n
  \subseteq F^n$ such that, for any $f\in H^n$, there is an element
  $g\in H^n$ for which ${(f\inter s) \union (g\inter \bar s) \notin
    F^n}$. Then there must exist a minimal integer $2\leq \lambda\leq
  |H^n|$ and $\lambda$~elements~$\{f_i \mid 1\leq i\leq \lambda\}$
  of~$H^n$ such that
  \[
  \forall 1\leq i \leq \lambda -1 \ (f_i\inter s) \union
  (f_{i+1}\inter \bar s) \not\in F^n \text{ and } (f_\lambda\inter
  s) \union (f_1\inter \bar s) \not\in F^n.
  \]
  By~Lemma~\ref{lemma-compl}, the complement of~$F^n$ is
  semi-stable. Hence, considering $(f_{\lambda-1}\inter
  s)\union(f_{\lambda}\inter\bar s)$ and $(f_\lambda\inter
  s)\union(f_1\inter\bar s)$, one of the following two vectors is not
  in~$F^n$:
  \begin{alignat*}1
        \bigl([(f_{\lambda-1}\inter s)\union(f_\lambda\inter\bar
          s)]\inter s\bigr) \union \bigl([(f_\lambda\inter
          s)\union(f_1\inter\bar s)]\inter \bar s\bigr)
        \\ \bigl([(f_\lambda\inter s)\union(f_1\inter\bar s)]\inter
        s\bigr) \union \bigl([(f_{\lambda-1}\inter
          s)\union(f_\lambda\inter\bar s)]\inter \bar s\bigr)
  \end{alignat*}
  The second expression equals~$f_\lambda$, which is in~$F^n$. Hence
  we get that $(f_{\lambda-1}\inter s) \union (f_1\inter \bar s)$ is
  not in~$F^n$, contradicting minimality of~$\lambda$.
\end{proof}

\dkfzefaaaaze*

\begin{proof}
We~begin with the first statement.  Assume that $F^n$ is semi-stable,
and take $f'=\flip_b(f)$ and $g'=\flip_b(g)$ in~$\flip_b(F^n)$, and
$s\in\ZO^n$. Then
\begin{alignat*}1
  (f'\inter s) \union (g'\inter\bar s) &= 
  \left(
  ((f\inter b) \union (\bar f \inter \bar b))\inter s
  \right) \union \left(
  ((g\inter b) \union (\bar g \inter \bar b))\inter \bar s
  \right)  \\
  &=
  \left(
  ((f\inter s) \union (g \inter \bar s))\inter b
  \right) \union \left(
  ((\bar f\inter s) \union (\bar g \inter \bar s))\inter \bar b
  \right)
\end{alignat*}
Write $\alpha=(f\inter s) \union (g \inter \bar s)$ and $\beta = (\bar
f\inter s) \union (\bar g \inter \bar s)$. One can easily check that
$\beta=\bar\alpha$. We~then have
\begin{alignat}1
  (f'\inter s) \union (g'\inter\bar s) &= 
  \left(
  \alpha\inter b
  \right) \union \left(
  \bar\alpha \inter \bar b
  \right)\notag \\
  &= \flip_b(\alpha).\label{eq-flip}
\end{alignat}
This computation being valid for any~$f$ and~$g$, we~also have
\begin{alignat}1
  (g'\inter s) \union (f'\inter\bar s) &= 
  \left(
  \gamma \inter b
  \right) \union \left(
  \bar\gamma \inter \bar b
  \right) \notag  \\
  &= \flip_b(\gamma) \label{eq-flip::Bis::AjoutFevrier}
\end{alignat}
with $\gamma=(g\inter s) \union (f \inter \bar s)$. By~hypothesis, at
least one of~$\alpha$ and~$\gamma$ belongs to~$F^n$, so that also at
least one of $(f'\inter s) \union (g'\inter\bar s)$ and $(g'\inter s)
\union (f'\inter\bar s)$ belongs to $\flip_b(F^n)$.

\medskip

    The second statement of Lemma~\ref{FLipLemmaPropiSLEG::dzal}
    trivially holds for $F^n=\emptyset$; thus in the following, we
    assume $F^n$ to be non-empty.  For~$1\leq i\leq n$,
    let~$s_i\in\ZO^n$ be the vector such that $s_i(j)=1$ if, and
    only~if, $j=i$. Applying Lemma~\ref{lemma-H}, we~get that for
    any~$i$, there exists some~$f_i\in F^n$ such that for any~$f\in
    F^n$, it~holds
    \begin{equation}
      (f_i\inter s_i)\union (f\inter \bar s_i) \in F^n. 
      \label{eq-i}
    \end{equation}
    We fix such a family $(f_i)_{i\leq n}$ then define~$g\in\ZO^n$ as
    $g=\bigcurlyvee_{1\leq i\leq n} (f_i \inter s_i)$,
    i.e. $g(i)=f_i(i)$ for all~$1\leq i\leq n$.  Starting from any
    element of~$F^n$ and applying Equation~\eqref{eq-i} iteratively
    for each~$i$, we~get that $g\in F^n$. Since $g\inter s_i=f_i\inter
    s_i$, we~also have
    \[
    \forall f\in F^n \qquad (g\inter s_i)\union (f \inter \bar s_i)\in F^n
    \]
    By~Equation~\eqref{eq-flip::Bis::AjoutFevrier}, since
    $\flip_g(g)=\bfone$, we~get
    \begin{equation}
      \forall f\in F^n \qquad  (\bfone \inter s_i)\union (\flip_g(f)\inter \bar s_i) \in \flip_g(F^n).
      \label{eq-1}
    \end{equation}
    Now, assume that $\flip_g(F^n)$ is not upward closed: then there exist elements
    $f\in F^n$ and $h\notin F^n$ such that $\flip_g (f)(i)=1\Rightarrow \flip_g (h)(i)=1$ for all~$i$. 
    Starting from~$f$ and iteratively applying Equation~\eqref{eq-1} for those~$i$
    for which $\flip_g (h)(i)=1$ and $\flip_g (f)(i)=0$, we~get that $\flip_g (h)\in \flip_g (F^n)$ and $h\in F^n$. Hence $\flip_g(F^n)$
    must be upward closed.
  \end{proof}

\subsection{Proof of Lemma~\ref{lemmaOrderingFsh}}
\label{app-lemma18}
\RTOSNE*

\begin{proof}
 Assume otherwise, there is $h'_1\in  \mathbb{F}^n(h_1,s)\backslash \mathbb{F}^n(h_2,s)$ and $h'_2\in  \mathbb{F}^n(h_2,s)\backslash \mathbb{F}^n(h_1,s)$. We then have:
  \begin{alignat*}1
    (h_1\inter s)\union (h_1'\inter\overline{s})\in F^n \quad & \quad
    (h_2\inter s)\union (h_1'\inter \overline{s}) \not\in F^n \\
    (h_2 \inter s)\union (h_2'\inter \overline{s}) \in F^n \quad & \quad
    (h_1\inter s)\union (h_2\inter \overline{s}) \not\in F^n
  \end{alignat*}
  
  Now consider $(h_1\inter s)\union (h_1'\inter \overline{s})$, $(h_2\inter s)\union (h_2'\inter \overline{s})$ and $s$. As $F^n$ is semi-stable, one of the two following vector is in $F^n$ :
  \begin{alignat*}2
    &\big( (h_1\inter s)\union (h_1'\inter \overline{s}) \inter s\big) \union \big( (h_2\inter s)\union (h_2'\inter \overline{s})\inter \overline{s} \big) \\
    &\big( (h_2\inter s)\union (h_2'\inter \overline{s}) \inter s\big) \union \big( (h_1\inter s)\union (h_1'\inter \overline{s})\inter \overline{s} \big) 
  \end{alignat*}
  The first vector is equal to $(h_1\inter s)\union (h_2'\inter \overline{s})$ and the second to $(h_2\inter s)\union (h_1'\inter \overline{s})$ and both are supposed to be in $\overline{F^n}$, we get a contradiction.
\end{proof}

\smallskip
While it is not related to the lemma above, we~prove here a result
that will be useful for the proof of
Lemma~\ref{lemmaintermediaireavantElemeoptim} in
Appendix~\ref{app-thm13}.

\begin{lemma}\label{LemmaOrderGetCoarserwithS}
  Given a \semist set $F^n$, $s_1,s_2\in\ZO^n$ such that $s_1 \inter
  s_2 =\bfzero$ and $f,g\in \ZO^n$ such that $f\preceq_{s_1} g$ and
  $f\preceq_{s_2} g$. Then $f\preceq_{s_1\union s_2} g$.
\end{lemma}

\begin{proof}
  Because $f\preceq_{s_1} g$ and $f\preceq_{s_2} g$, we have
  \begin{alignat}1
    \forall i\in\{1,2\}~\forall h\in\ZO^n \qquad (f\inter s_i) \union
    (h \inter \overline{s_i}) \in F^n \Rightarrow (g\inter s_i) \union
    (h \inter \overline{s_i}) \in F^n
    \label{OrderingEG::DetailleLemma::sOne::sTwo}
  \end{alignat}
  Consider $h'\in \ZO^n$ such that $\alpha=(f\inter (s_1\union s_2))
  \union (h' \inter \overline{(s_1\union s_2)})$ is in $F^n$.  Define
  the element $h= \alpha \inter \overline{s_2}$, then $ (f\inter s_2)
  \union (h \inter \overline{s_2}) = (f\inter (s_1\union s_2)) \union
  (h' \inter \overline{(s_1\union s_2)}) \in F^n $.
  Using~\eqref{OrderingEG::DetailleLemma::sOne::sTwo} with $s_2$
  and~$h$, we get $ \beta = (g\inter s_2) \union (h \inter
  \overline{s_2}) $.  As $s_1 \inter s_2 =\bfzero$, we can write
  $\beta= (f\inter s_1) \union (g\inter s_2) \union (h' \inter
  \overline{(s_1\union s_2)}) \in F^n $.
      
  Now consider $h= \beta \inter \overline{s_1}$, we have $(f\inter
  s_1)\union (h\inter \overline{s_1})=\beta \in
  F^n$. Using~\eqref{OrderingEG::DetailleLemma::sOne::sTwo} with $s_1$
  and h, we get $(g \inter (s_1 \union s_2)) \union (h' \inter
  \overline{(s_1\union s_2)}) \in F^n$.  Therefore
  $\mathbb{F}^n(f,s_1\union s_2)\subseteq \mathbb{F}^n(g,s_1 \union
  s_2)$ and $f\preceq_{s_1\union s_2} g$.
\end{proof}

\subsection{Proof of Theorem~\ref{ThmModelCheck}}
\label{app-thm13}
\MainTheoremeDuChapSLEGPartII*

\begin{proof}

  Following Lemma~\ref{lemma-upwardcl}, we~assume for the rest of the
  proof that the set $F^n$ of the \SLEGf formula $\phi$ is upward
  closed (even if it means negating some of the \LTL objectives).
  We~also assume it is non-empty, since the result is trivial otherwise.
  
  We~then have the following property:
  \begin{lemma}\label{lemma-upwcl}
    Assuming $F^n$ is upward-closed, for any $f$, $g$ and~$s$
    in~$\ZO^n$, if $f\leq g$ (i.e. for all~$i$, $f(i)=1 \impl
    g(i)=1$), then $f \preceq_s g$.  In~particular, $\bfzero$~is a
    minimal element for~$\preceq_s$, for~any~$s$.
  \end{lemma}     

  \begin{proof}
    Since $f\leq g$, then also $(f\inter s)\union (h\inter \bar s)
    \leq (g\inter s)\union (h\inter\bar s)$, for any~$h\in\ZO^n$.
    Since $F^n$ is upward-closed, if $(f\inter s)\union (h\inter \bar
    s)$ is in~$F^n$, then so is $(g\inter s)\union (h\inter \bar s)$.
  \end{proof}

  We now develop the proof of Theorem~\ref{ThmModelCheck}. The proof
  is in three steps:
  \begin{itemize}
    \item we build a family of parity automata expressing the
      objectives that may have to be fulfilled along
      outcomes. A~configuration is then described by a state~$q$ of
      the game, a~vector~$d$ of states of those parity automata, and a
      set~$s$ of goals that are still \emph{active} along the current
      outcome;
    \item we define formulas encoding the two ways of
      fulfilling a set of goals: either by fulfilling all goals along
      the same outcome, or by partitioning them among different
      branches;
    \item by~turning the formulas above into 2-player parity games,
      we~inductively compute optimal sets of goals (represented as
      vectors~$b_{q,d,s}\in\ZO^n$) that can be achieved from a given
      configuration and for each subset of active
      goals. By~determinacy of parity games, we~derive timeline maps
      witnessing the fact that $b_{q,d,s}$ can be achieved, and the
      fact that it is optimal. If $b_{q_0,d_0,\bfone} \in F^n$, we~get
      a witness map for $\calG,q_0\models^T\phi$; otherwise, we~get
      one for $\calG,q_0\models^T\non\phi$.
  \end{itemize}
  

  \subsubsection{Automata for conjunctions of goals}
  
  We~use \emph{deterministic parity word automata} to keep track of
  the goals to be satisfied. Since we initially have no clue about
  which goal(s) will have to be fulfilled along an outcome, we~use a
  (large) set of automata, all running in parallel.

  For~$s\in\ZO^n$ and $h\in\ZO^n$, we~let
  $D_{s,h}$ be a deterministic parity automaton accepting exactly the words
  over~$2^{\AP}$ along which the following 
  formula~$\Phi_{s,h}$ holds:
  \begin{xalignat*}1
    \Phi_{s,h} = \bigvee_{\substack{k\in \ZO^n \\ h~\preceq_s ~k}} 
    \bigwedge_{\substack{j\text{ s.t.}\\(k \inter
	s)(j)=1}} \formuleLTL_j. 
  \end{xalignat*}
  where a conjunction over an empty set (i.e., if~$(k\inter s)(j)=0$
  for all~$j$) is true.  Notice that, using Lemma~\ref{lemma-upwcl},
  if~$h\preceq_s k$ and $k\leq k'$, then $h\preceq_s k'$, so that
  we~do not need to enforce $\not\phi_j$ for those indices where
  $(k\inter s)(j)=0$.


  As an example, take~$s\in\ZO^n$ with $\size s=1$, writing~$j$ for
  the index where $s(j)=1$; for any $h\in\ZO^n$, if there is
  $k\succeq_s h$ with $k(j)=0$ (which in particular is the case when
  $h(j)=0$), then the~automaton~$D_{s,h}$ is universal; otherwise
  $D_{s,h}$~accepts the set of words over~$2^{\AP}$ along which
  $\formuleLTL_j$ holds.

  We write $\calD= \{D_{s,h} \mid s\in\ZO^n,\ h\in\ZO^n\}$ for the set
  of automata defined above. A~\emph{vector of states of~$\calD$} is
  a~function associating with each automaton~$D\in\calD$ one of its
  states.  We~write~$\VS$ for the set of all vectors of states
  of~\calD.  For~any vector $d\in\VS$ and any state~$q$ of~$\calG$,
  we~let $\succs(d,q)$ to be the vector of states
  associating with each $D\in\calD$ the successor of state~$d(D)$ after
  reading~$\labels(q)$; we~extend~$\succs$
  to finite paths~$(q_i)_{0\leq i\leq n}$ in~$\calG$ inductively, letting
  $\succs(d,(q_i)_{0\leq i\leq n}) = \succs(\succs(d,(q_i)_{0\leq i\leq n-1}),q_n)$.

  An infinite path~$(q_i)_{i\in \bbN}$ in \calG is
  accepted by an automaton~$D$ of~\calD whenever the word
  $(\labels(q_i))_{i\in\bbN}$ is accepted by~$D$.  We~write~$\calL(D)$ for
  the set of paths of~$\calG$ accepted by~$D$.
  Finally, for $d\in\VS$, we~write~$\calL (D^d_{s,h})$ for the set of
  words that are accepted by $D_{s,h}$ starting from the state
  $d(D_{s,h})$ of~$D_{s,h}$.
  
  \begin{proposition}
    \label{propositionlanguageDinclusion}\label{prop-Dsh}
    The following holds for any $s\in\ZO^n$:
    \begin{enumerate}
      \item $\Phi_{s,\bfzero} \equiv \top$ (i.e., $D_{s,\bfzero}$ is universal);
      \item for any $h_1,h_2\in\ZO^n$, if $h_1\preceq_s h_2$, we have
        $\Phi_{s,h_2} \impl \Phi_{s,h_1}$ (i.e., 
        $\calL (D_{s,h_2}) \subseteq \calL(D_{s,h_1})$);
      \item for any $h\in F^n$, $\Phi_{\bfone,h}\equiv
        \bigvee_{k\in F^n}\bigwedge_{j\text{ s.t. }
          k(j)=1}\ \formuleLTL_j$. 
    \end{enumerate}
  \end{proposition}
  \begin{proof}
    $\Phi_{s,\bfzero}$ contains the empty conjunction ($k=\bfzero$) as
    a disjunct. Hence it~is equivalent to~true. When~$h_1\preceq_s
    h_2$, formula~$\Phi_{s,h_1}$ contains more disjuncts
    than~$\Phi_{s,h_2}$, hence the second result. Finally,
    $\bbF^n(f,\bfone)=\ZO^n$ if~$f\in F^n$, and is empty
    otherwise. Hence if $h\in F^n$, we~have $h\preceq_\bfone k$ if,
    and only~if, $k\in F^n$, which entails the result.
  \end{proof}
  
  \subsubsection{Two ways of achieving goals}
  After a given history, a set of goals may be achieved either along a
  single outcome, in case the assignment of strategies to players
  gives rise to the same outcomes, or they may be split among different
  outcomes. We~express those two ways of satisfying goals,
  by means of two operators parameterized by the current configuration.


  The first operator covers the case where the goals currently enabled
  by~$s$ (those goals $\beta_i\phi_i$ for which $s(i)=1$) are all
  fulfilled along the same outcome.
  For any $d \in \textsf{VS}$ and any two $s$ and $h$ in $\ZO^n$, the
  operator~$\Gamma^{\stick}_{d,s,h}$ is defined as follows: given a
  context~$\chi$ with~$\calV \subseteq\dom\chi$ and a state~$q$
  of~\calG,
  \[
    \calG,q \models_{\chi} \Gamma^{\stick}_{d,s,h} 
    \ \iff\ 
      \exists \rho\in\Play_\calG(q) \text{ s.t. }
    \begin{cases}
      \text{--}\ \forall j\leq n.\
      \bigl(s(j)=1 \impl \out(q,\chi\circ\beta_j)=\rho\bigr)
      \\ 
      \text{--}\ \rho \in \calL(D_{s,h}^d)
    \end{cases}
  \]
  Intuitively, all the goals enabled by~$s$ must
  give rise to the same outcome, which is accepted by~$D_{s,h}^d$.
  
  \medskip
  
  We~now consider the case where the active goals 
  are partitioned among different outcomes.%
  \begin{definition}
    A~\emph{partition} of an element $s\in \ZO^n$ is a sequence
    $(s_\kappa)_{1\leq \kappa \leq \lambda}$, with $\lambda\geq 2$, of elements
    of $\ZO^n$ with $s_1\union \dots \union s_\lambda = s$ and where
    for any two $\kappa \neq \kappa' $ and any $1\leq j\leq n$, we have
    $s_\kappa(j)=1 \impl s_{\kappa'}(j)=0$.
    
    An \emph{extended} partition of $s$ is a sequence $\tau=(s_\kappa,
    q_\kappa, d_\kappa )_{1\leq\kappa\leq\lambda}$ of elements of
    $\ZO^n\times \Q \times \textsf{VS}$ where $(s_\kappa)_{1\leq \kappa \leq
      \lambda}$ is a partition of~$s$, $q_\kappa$ are states
    of~$\calG$, and $d_\kappa$ are vectors of states of the automata
    in~$\calD$.
  \end{definition}
  
  We write $\Part(s)$ for the set of all extended partitions
  of~$s$.  Notice that we only consider non-trivial partitions;
  in~particular, if $\size{s}\leq 1$, then
  $\Part(s)=\emptyset$.
  For any $d\in \textsf{VS}$, any $s$ in $\ZO^n$ and any set of
  partitions $\Upsilon_s$ of~$s$, the~operator
  $\Gamma^{\sep}_{d,s,\Upsilon_s}$ states that the goals currently
  enabled by~$s$ all follow a common history~$\rho$ for a finite
  number of steps, and then partition themselves according to some
  partition in~$\Upsilon_s$.  The semantics
  of~$\Gamma^{\sep}_{d,s,\Upsilon_s}$ is defined as follows:%
  \[
    \calG,q \models_{\chi} \Gamma^{\sep}_{d,s,\Upsilon_s}
    \ \iff \
    \begin{array}{l}
      \exists \tau\in \Upsilon_s.\\ \exists \rho\in \Hist_\calG(q).\
    \end{array}
    \begin{cases}
      \text{--}\ \forall j\leq n.\ 
      \bigl(s(j)=1 \impl \rho \in \Pref{}(\out(q,\chi\circ\beta_j))\bigr)
      \\ 
      \text{--}\ \forall  \kappa \leq \size{\tau}.\ \forall  j\leq n.\
      \text{ letting } m_{j}(A)=\chi(\beta_j(A))(\rho).\\  
      \phantom{\text{--}\ }\bigl(s_\kappa(j)=1 \impl
        q_\kappa=\Delta(\lst{\rho}, m_j)\bigr) \\
      \text{--}\ \forall \kappa \leq \size{\tau}.\
      \succs(d,\rho\cdot q_\kappa)=d_\kappa.
    \end{cases}
    \]
    Notice that $h$ does not appear explicitly in this definition, but
    $\Gamma^{\sep}_{d,s,\Upsilon_s}$ will depend on~$h$ through the
    choice of~$\Upsilon_s$. The operators $\Gamma^{\stick}$ and
    $\Gamma^{\sep}$ are illustrated on Fig.~\ref{fig-Gamma}.
  
    \begin{figure}[t]
      \centering
      \begin{tikzpicture}
        \begin{scope}
          \draw (0,0) node[fill,circle,inner sep=0pt,minimum size=4pt] (a) {}
            node[right] {$q\models_\chi \Gamma^{\stick}_{d,s,h}$};
          \draw (a) .. controls +(-100:1cm) and +(70:1cm) .. +(0,-2);
          \draw[dashed,-latex'] (0,-2) .. controls +(-110:5mm)
            and +(100:5mm) .. (0,-3)  node[right=2mm] {$\in \calL(D_{s,h}^d)$};
          \draw[-latex',shorten >=5mm,shorten <=5mm] (.4,0) .. controls +(-100:1cm) and +(70:1cm) .. ++(0,-2) .. controls +(-110:5mm) and +(100:5mm) .. ++(0,-1);
          \path (2.5,-1.4) node[text width=3.4cm] {same outcome for all goals enabled by~$s$};
        \end{scope}
        \begin{scope}[xshift=7cm]
          \draw (0,0) node[fill,circle,inner sep=0pt,minimum size=4pt] (a) {}
          node[right] {$q\models_{\chi} \Gamma^{\sep}_{d,s,\Upsilon_s}$};
          \draw[-latex'] (a) .. controls +(-100:1cm) and +(70:1cm) .. +(0,-2)
          node[fill,circle,inner sep=0pt,minimum size=4pt] (b) {};
          \path (2.2,-1) node[text width=3.4cm] {common history~$\rho$  for all goals enabled by~$s$};
          \begin{scope}
          \everymath{\scriptstyle}
          \draw (-1,-2.8) node[fill,circle,inner sep=0pt,minimum size=4pt] (c) {}
          node[below] {$q_1$, $d_1$} node[below=3mm] {$s_1$};
          \draw (0,-2.8) node[fill,circle,inner sep=0pt,minimum size=4pt] (d) {}
          node[below] {$q_2$, $d_2$} node[below=3mm] {$s_2$};
          \draw (1,-2.8) node[fill,circle,inner sep=0pt,minimum size=4pt] (e) {}
          node[below] {$q_3$, $d_3$} node[below=3mm] {$s_3$};
          \draw[-latex'] (b) -- (c);
          \draw[-latex'] (b) -- (d);
          \draw[-latex'] (b) -- (e);
          \end{scope}
          \begin{scope}[rounded corners=3mm]
          \draw (-1,-3.1) node[carre,dotted,minimum width=10mm] (s1) {};
          \draw (0,-3.1) node[carre,dotted,minimum width=9mm] (s2) {};
          \draw (1,-3.1) node[carre,dotted,minimum width=11mm] (s3) {};
          \end{scope}
          \draw (2.4,-3) node (t) {$\in \Upsilon_s$};
          \draw[dotted] (s1.-90) .. controls +(-20:25mm) and +(-130:10mm) .. (t);
          \draw[dotted,-latex'] (s2.-90) .. controls +(-25:10mm) and +(-130:10mm) .. (t);
          \draw[dotted] (s3.-90) .. controls +(-30:4mm) and +(-130:10mm) .. (t);
        \end{scope}
      \end{tikzpicture}
      \caption{Illustration of $\Gamma^{\stick}_{d,s,h}$ and
        $\Gamma^{\sep}_{d,s,\Upsilon_s}$}
      \label{fig-Gamma}
    \end{figure}
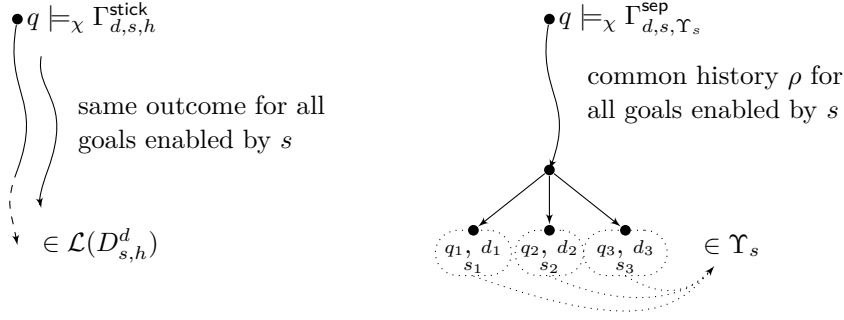

  \subsubsection{Fulfilling optimal sets of goals}
  We~now inductively (on~$\size s$) define new operators~$\Gamma_{d,s,h}$
  combining the above two operators $\Gamma^{\stick}$
  and~$\Gamma^{\sep}$, and selecting optimal ways of partitioning the
  goals among the outcomes.

  \paragraph*{Base case: $\size s=1$.}
  When only one goal is enabled, we~only have to consider a single
  outcome, so that we~let $\Gamma_{d,s,h}=\Gamma^{\stick}_{d,s,h}$,
  for any~$d\in\VS$ and~$h\in\ZO^n$. 
  By~Prop.~\ref{prop-Dsh}, for any valuation~$\chi$ such that
  $\Agt\subseteq \dom\chi$, it~holds $\calG,q\models_\chi
  \Gamma_{d,s,\bfzero}$, hence also $\calG,q\models^T (Q_i x_i)_{1\leq
    i\leq l}.\ \Gamma_{d,s,\bfzero}$.  Hence there must exist a
  maximal value~$b\in\ZO^n$ such that $\calG,q\models^T (Q_i
  x_i)_{1\leq i\leq l}.\ \Gamma_{d,s,b}$. We~write~$b_{q,d,s}$ for one
  such value (notice that it need not be unique).
  By~maximality, for any~$h$ such that $b_{q,d,s}\prec_s h$, we~have
  $\calG,q\not\models^T (Q_i  x_i)_{1\leq i\leq l}.\ \Gamma_{d,s,h}$.

  \paragraph*{Induction step.}
  We~assume that for any~$d\in\VS$, any $h\in\ZO^n$ and
  any~$s\in\ZO^n$ with $\size s\leq k$, we~have defined an
  operator~$\Gamma_{d,s,h}$, and that for any~$q\in \Q$, we~have fixed
  an element~$b_{q,d,s}\in\ZO^n$ for which $\calG,q\models^T (Q_i
  x_i)_{1\leq i\leq l}.\ \Gamma_{d,s,b}$ and such that for any~$h$
  such that $b_{q,d,s}\prec_s h$, it~holds $\calG,q\not\models^T (Q_i
  x_i)_{1\leq i\leq l}.\ \Gamma_{d,s,h}$. 

  Pick~$s\in\ZO^n$ with $\size s=k+1$, and an extended
  partition~$\tau=(s_\kappa,q_\kappa,d_\kappa)_{1\leq\kappa\leq\lambda}$. Then
  we~must have $\size{s_\kappa}<k+1$ for all~$1\leq\kappa\leq\lambda$,
  so that $\Gamma_{d_\kappa,s_\kappa,h}$ and
  $b_{q_\kappa,d_\kappa,s_\kappa}$ have been defined at previous
  steps.  We~let
  \[
  c_{s,\tau} = \bigcurlyvee_{1\leq\kappa\leq\lambda} (s_\kappa \inter
  b_{q_\kappa,d_\kappa,s_\kappa}).
  \]
  We~then define
  \[
  \Gamma_{d,s,h} = \Gamma^{\stick}_{d,s,h} \ou \Gamma^{\sep}_{d,s,\Upsilon_{s,h}}
  \qquad\text{ with } \Upsilon_{s,h} = \{\tau\in\Part(s) \mid
  h\preceq_s c_{s,\tau} \}.
  \]

  As~previously, we~claim that $\calG,q\models_\chi
  \Gamma_{d,s,\bfzero}$ for any $\chi$ such that $\Agt\subseteq
  \dom\chi$. Indeed, for a given~$\chi$, if~all the outcomes of the
  goals enabled by~$s$ follow the same infinite path, then this path
  is accepted by~$D_{s,\bfzero}$ and $\calG,q\models_\chi
  \Gamma^{\stick}_{d,s,\bfzero}$; otherwise, after some common
  history~$\rho$, the~outcomes are partitioned following some extended
  partition~$\tau_0$, which obviously satisfies $\bfzero \preceq_s
  c_{s,\tau_0}$ since~$\bfzero$ is a minimal element
  of~$\preceq_s$. Hence in that case $\calG,q\models_\chi
  \Gamma^{\sep}_{d,s,\Upsilon_{s,\bfzero}}$.

  In~particular, it~follows that $\calG,q\models^T (Q_i x_i)_{1\leq
    i\leq l}.\ \Gamma_{d,s,\bfzero}$, and we can fix a maximal
  element~$b_{q,d,s}$ for which $\calG,q\models^T (Q_i x_i)_{1\leq
    i\leq l}.\ \Gamma_{d,s,b_{q,d,s}}$ and $\calG,q\not\models^T (Q_i
  x_i)_{1\leq i\leq l}.\ \Gamma_{d,s,h}$ for any~$h\succ_s b_{q,d,s}$.

  \medskip

  This concludes the inductive definition
  of~$\Gamma_{d,s,b_{q,d,s}}$. We~now prove that it satisfies the
  following lemma:
  \begin{lemma}
    \label{LemmaexistsmapsvarrhosProofThmModelCheck}
    \label{lemma-tmap}
    For any~$q\in \Q$, any $d\in\VS$ and any~$s\in\ZO^n$,
    \begin{itemize}
    \item there exists a timeline map~$\vartheta_{q,d,s}$ for~$(Q_i
      x_i)_{1\leq i\leq l}$ witnessing the fact that
      \[
      \calG, q\models^T (Q_i x_i)_{1\leq i\leq l}.\ \Gamma_{d,s,b_{q,d,s}}
      \]
    \item for any~$h\succ_s b_{q,d,s}$, there exists a timeline map~$\bar\vartheta_{q,d,s,h}$ for~$(\bar Q_i
      x_i)_{1\leq i\leq l}$ witnessing the fact that
      \[
      \calG, q\models^T (\bar Q_i x_i)_{1\leq i\leq l}.\
        \non \Gamma_{d,s,h}.
        \]
    \end{itemize}
  \end{lemma}

  \begin{proof}
    The first result is a direct consequence of the
    construction. To~prove the second part, we~again turn the
    satisfaction of~$\Gamma_{d,s,h}$, for~$h\succ_s b_{q,d,s}$, into a
    parity game, as for the proof of Prop.~\ref{prop-negSL1G}. We~only
    sketch the proof here, as it involves the same ingredients.

    The parity game is obtained from~$\calG$ by replacing each state
    by a quantification game.  We~also introduce two sink states,
    $q_{\even}$ and~$q_{\odd}$, which respectively are winning for
    Player~$P_\exists$ and for Player~$P_{\forall}$.  When arriving at
    a leaf~$(q,d,\frakm)$ of the $(q,d)$-copy of the quantification
    game, there may be one of the following three transitions
    available:
    \begin{itemize}
    \item if there is a state~$q'$ such that for all $j$ with $s(j)=1$,
      it~holds $q'=\Delta(q, m_{\beta_j})$ (in~other terms, the moves
      selected in the current quantification game generate the same
      transition for all the goals enabled by~$s$), then there is a
      single transition to~$(q',d',\epsilon)$, where $d'=\succs(d,q')$.
    \item otherwise, if there is an extended
      partition~$\tau=(s_\kappa,q_\kappa,d_\kappa)_{1\leq\kappa\leq\lambda}$
      of~$s$ such that $c_{s,\tau} \succeq_s h$ and, for
      all~$1\leq\kappa\leq\lambda$, for all~$j$ such that
      $s_\kappa(j)=1$, we~have $\Delta(q,m_{\beta_j})=q_\kappa$ and
      $\succs(d,q_\kappa)=d_\kappa$, then there is a transition
      from~$(q,d,\frakm)$ to~$q_{\even}$.
    \item otherwise, there is a transition from~$(q,d,\frakm)$
      to~$q_{\odd}$.
    \end{itemize}
    The priorities defining the parity condition are inherited from
    those in~$D_{s,h}$.

    Since $\calG,q\not\models^{T} (Q_i x_i)_{1\leq i\leq l}.\
    \Gamma_{d,s,h}$, Player~$P_\exists$ does not have a winning
    strategy in this game, and by determinacy Player~$P_\forall$
    has~one. From the winning strategy of Player~$P_\forall$,
    we~obtain a timeline map~$\bar\vartheta_{q,d,s,h}$ for~$(\bar Q_i
    x_i)_{1\leq i\leq l}$ witnessing the fact that $\calG,q\models^T
    (\bar Q_i x_i)_{1\leq i\leq l}.\ \non\Gamma_{d,s,h}$.
  \end{proof}
  
  \begin{remark}
    While the definition of $\Gamma_{d,s,b_{q,d,s}}$ (and in
    particular of~$b_{q,d,s}$) is not effective, the parity games
    defined in the proof above can be used to compute each $b_{q,d,s}$
    and $\Gamma_{d,s,b_{q,d,s}}$. Indeed, such parity games can be
    used to decide whether $\calG,q\models^T (Q_i x_i)_{1\leq i\leq
      l}.\ \Gamma_{d,s,h}$ for all~$h$, and selecting a maximal value
    for which the result holds.
    Then $b_{q_0,d_0,\bfone}\in F^n$ implies $\calG,q_0\models^T (Q_i
    x_i)_{1\leq i\leq l} F^n(\beta_j\phi_j)_{1\leq j\leq n}$.

    Each parity game has size doubly-exponential, with
    exponentially-many priorities; hence they can be solved in
    \EXPTIME[2]. The number of games to solve is also
    doubly-exponential, so that the whole algorithm runs in
    \EXPTIME[2].
%
%
  \end{remark}

  We~now focus on the operator obtained at the end of the induction,
  when~$s=\bfone$. Following Prop.~\ref{prop-Dsh},
  $\calL(D_{\bfone,f})$~does not depend on the exact value of~$f$, as
  soon as it is in~$F^n$. We~then let
  \[
  \Gamma_{F^n} = \Gamma^{\stick}_{d_0,\bfone,f} \ou \Gamma^{\sep}_{d_0,\bfone,\Upsilon_{F^n}}
  \]
  where $f$ is any element of~$F^n$ (remember $F^n$ is assumed to be
  non-empty), $d_0$ is the vector of initial states of the automata
  in~$\calD$, and $\Upsilon_{F^n}=\{\Part(\bfone) \mid c_{\bfone,\tau}
  \in F^n\}$.
  We~write $\vartheta_\bfone$ and $\bar\vartheta_\bfone$ for the maps
  $\vartheta_{q_0,d_0,\bfone}$ and $\bar\vartheta_{q_0,d_0,\bfone,h}$
  for some~$h\in F^n$, as given by Lemma~\ref{lemma-tmap}. From the
  discussion above, $\bar\vartheta_{q_0,d_0,\bfone,h}$ does not depend
  on the choice of~$h$ in~$F^n$, and we simply write~it
  $\bar\vartheta_{q_0,d_0,\bfone}$.

  Then:
  \begin{lemma}
    \label{LemmaBidexistsmapsvarrhosProofThmModelCheck}
    \label{lemma-rho1}
    If $\calG ,q_0\models^{T} (Q_i x_i)_{1\leq i\leq
      l}.\ \Gamma_{F^n}$, then $\vartheta_{\bfone}$ witnesses that $
    \calG,q_0 \models^{T} (Q_i x_i)_{1\leq i\leq l}.\ \Gamma_{F^n} $.
    Conversely, if $\calG ,q_0 \not\models^{T} (Q_i x_i)_{1\leq i\leq
      l}.\ \Gamma_{F^n}$, then $\bar\vartheta_{\bfone}$ witness
    that $ \calG,q_0 \models^{T} (\bar Q_i x_i)_{1\leq i\leq
      l}.\ \neg \Gamma_{F^n} $.
  \end{lemma}
  
  \begin{proof}
    The first part directly follows from the previous lemma. For the
    second part, $\calG ,q_0 \not\models^{T} (Q_i x_i)_{1\leq i\leq
      l}.\ \Gamma_{F^n}$ means that $b_{q_0,d_0,\bfone}\notin
    F^n$. Hence for any~$f\in F^n$, we~have $f\succ_s
    b_{q_0,d_0,\bfone}$, so that $\bar\vartheta_{q_0,d_0\bfone}$ is a
    witness that $\calG,q\models^T (\bar Q_i x_i)_{1\leq i\leq
      l}.\ \non\Gamma_{F^n}$.
  \end{proof}

  \subsubsection{Compiling optimal maps}
  From Lemma~\ref{lemma-tmap}, we~have timeline maps for each~$q$, $d$
  and~$s$. We~now compile them into two map~$\theta$ and~$\bar\theta$.
  The construction is inductive, along histories.

  Pick a history~$\rho$ starting from~$q_0$ and strategies for
  universally-quantified variables $w\colon \calV^\forall \to (\Hist
  \to \Act)$. Assuming $\theta$ has been defined along all strict
  prefixes of~$\rho$, a~goal~$\beta_j \phi_j$ is said \emph{active}
  after~$\rho$ w.r.t.~$\theta(w)$ if the following condition holds:
  \[
  \forall i<\size\rho.\ \rho(i+1) = \Delta(\rho(i),
  (\theta(w)(\beta_j(A))(\rho_{\leq i}))_{A\in\Agt}).
  \]
  In~other terms, $\beta_j \phi_j$ is active after~$\rho$
  w.r.t.~$\theta(w)$ if~$\rho$ is the outcome of strategies prescribed
  by~$\theta(w)$ under assignment~$\beta_j$. We~let
  $s_{\rho,\theta(w)}$ be the element of~$\ZO^n$ such that
  $s_{\rho,\theta(w)}(j)=1$ if, and only~if, $\beta_j \phi_j$ is
  active after~$\rho$ w.r.t.~$\theta(w)$.

  We~now define~$\theta(w)(x_i)(\rho)$ for all~$x_i\in\calV$:
  \begin{itemize}
  \item if~$x_i\in\calV^\forall$, we~let $\theta(w)(x_i)(\rho)= w(x_i)(\rho)$;
  \item if~$x_i\in\calV^\exists$, we consider two cases:
    \begin{itemize}
    \item if $s_{\rho,\theta(w)}=1$, then all goals are still active,
      and $\theta$ follows the map~$\vartheta_\bfone$:
      $\theta(w)(x_i)(\rho) = \vartheta_{\bfone}(w)(x_i)(\rho)$.
    \item otherwise, we~let~$\rho_1$ be the maximal prefix of~$\rho$
      for which $s_{\rho_1,\theta(w)} \not=
      s_{\rho,\theta(w)}$. We~may then write $\rho=\rho_1\cdot\rho_2$,
      and let~$q_1=\lst{\rho_1}$ and $d_1=\succs(d_0,\rho_1)$.
      We~then let
      $\theta(w)(x_i)(\rho)=\vartheta_{q_1,d_1,s_{\rho,\theta(w)}}
        (w_{\overrightarrow{\rho_1}})(x_i)(\rho_2)$.
    \end{itemize}
  \end{itemize}
  The dual map~$\bar\theta$ is defined in the same way, using
  maps~$\bar\vartheta$ in place of~$\vartheta$.

  The following result will conclude our proof of Theorem~\ref{ThmModelCheck}.
  \begin{restatable}{lemma}{lemmadeuxiemeannexethmmodelcheck}
    \label{lemmaintermediaireavantElemeoptim}  
    There exists a valuation~$\chi$ of domain~$\calV$ such that
    $\theta(\chi_{|\calV^\forall})=\chi$ and
    $\overline{\theta}(\chi_{|\calV^\exists})=\chi$. It~satisfies
    \begin{alignat*}3
      \calG ,q_0 &\models_{\chi} ~\Gamma_{F^n}
      &&\quad\impl\quad
      &\forall w\in(\Hist_\calG \to \Act)^{\calV^\forall}
      .\ \calG ,q_0 &\models_{\theta(w)}F^n(\beta_j\formuleLTL_j)_{1\leq j\leq n}
      \\
      \calG ,q_0 &\models_{\chi} ~ \neg \Gamma_{F^n}  
      &&\quad\impl\quad
      &\forall \overline{w}\in(\Hist_\calG \to \Act)^{\calV^\exists}
      .\ \calG ,q_0 &\models_{\overline{\theta}(\overline{w})} 
      \overline{F^n}(\beta_j\formuleLTL_j)_{1\leq j\leq n}
    \end{alignat*}
  \end{restatable}
  
  \begin{proof}
    We~use the same technique as in the proof of
    Prop.~\ref{PropNegationsAndCo::DeuxTroisAvril::ccr} (see
    Appendix~\ref{app-prop8}): from~$\theta$ and~$\bar\theta$,
    we~build a strategy valuation~$\chi$ on~$\calV$ such that
    $\theta(\chi_{|\calV^\forall})=\chi$ and
    $\overline{\theta}(\chi_{|\calV^\exists})=\chi$.

    We~introduce some more notations.
    For $w\colon \calV^\forall\to (\Hist_\calG \to \Act)$,  we~let
    \begin{itemize}
      \item $\pi^w_j$ be the outcome
        $\out(q_0,(\theta(w)((\beta_j(A))_{A\in\Agt}))$ for all~$1\leq
        j\leq n$;
      \item $f^w$ be the element of~$\ZO^n$ such that $f^w(j)=1$ if,
        and only~if, $\pi^w_j\models \phi_j$;
      \item $R^w \subseteq \ZO^n\times \Hist_\calG$ be the
        relation such that $(s,\rho)\in R^w$ if, and only~if,
        $s=s_{\rho,\theta(w)}$ and $\rho$ is minimal (meaning for any
        strict prefix $\rho'$ of~$\rho$, it~holds $(s,\rho')\notin
        R^w$).
    \end{itemize}
    
    \begin{lemma}
      \label{lemma-FinitionThmModelCheckOptimal}
      For any $w\colon \calV^\forall \to (\Hist_\calG \to \Act)$ and
      any~$\rho\in\Hist$, letting $d_\rho= \succs (d_0,\rho)$,
      it~holds
      \[
      \forall s\in\ZO^n.\ 
      (s,\rho)\in R^w
      \impl
      b_{\lst{\rho},d_\rho,s} \preceq_s f^w.
      \] 
    \end{lemma}

    \begin{proof}
      Fix some $w\in(\Hist_\calG \to \Act)^{\calV^\forall}$.
      The proof proceeds by induction on~$\size s$.
      
      \subparagraph*{Base case: $\size{s}=1$.}
      Assume $(s,\rho)\in R^w$.
      As $\size{s}=1$,
      there is a unique goal, say~$\beta_{j_0}\phi_{j_0}$,
      active along~$\rho$ w.r.t.~$\theta(w)$.
      By~definition of~$\theta$, $\pi_{j_0}=
      \rho\cdot \eta$ where $\eta$~is the outcome
      of $\vartheta_{\lst{\rho},d_\rho, s}(w_{\overrightarrow
        \rho})((\beta_j(A))_{A\in\Agt})$ from~$\lst{\rho}$.
      
      
      Because $\size{s}=1$, we~have
      $\Gamma_{d_\rho,s,b_{\lst{\rho},d_\rho, s}} =
      \Gamma^{\stick}_{d_\rho,s,b_{\lst{\rho},d_\rho, s}}$.
      The~map~$\vartheta_{\lst{\rho},d_\rho, s}$ is a witness that
      $\calG,\lst\rho \models^{T} (Q_i x_i)_{1\leq i\leq l}
      \Gamma_{d_\rho,s,b_{\lst{\rho},d_\rho, s}}$; therefore it also
      witnesses that $\calG,\lst\rho \models^{T} (Q_i x_i)_{1\leq
        i\leq l} \Gamma^{\stick}_{d_\rho,s,b_{\lst{\rho},d_\rho, s}}$.
      By~definition of the $\Gamma^{stick}$ operators, this implies
      that for any~$w$, the outcome of~$\vartheta_{\lst{\rho},d_\rho,
        s}(w_{\overrightarrow\rho})$ from~$\lst\rho$ is accepted by
      the automaton $D^{d_\rho}_{s,b_{\lst{\rho},d_\rho,s}}$;
      in~particular, $\eta$~is accepted by
      $D^{d_\rho}_{s,b_{\lst{\rho},d_\rho,s}}$.
      
      \medskip

      The automaton $D^{d_\rho}_{s,b_{\lst{\rho},d_\rho,s}}$ accepts
      paths which give better results (w.r.t.~$\preceq_s$) for the
      objectives $(\beta_j\varphi_j)_{j \mid s(j)=1}$ than
      $b_{\lst{\rho},d_\rho,s}$.
      In other terms, we~have $b_{\lst{\rho},d_\rho,s} \preceq_s f^w$.
      
      \subparagraph*{Induction step.}
      We assume that the
      Proposition~\ref{lemma-FinitionThmModelCheckOptimal} holds for
      any elements~$s\in\ZO^n$ of size $\size{s}<\alpha$. We~now
      consider for the induction step an element~$s\in\ZO^n$ such that
      $\size s=\alpha$ and $(s,\rho)\in R^w$.
      \begin{itemize}
      \item if the enabled goals all follow the same outcome, i.e., if
        there exists an infinite path~$\eta$ such that
        $\pi_j=\rho\cdot\eta$ for all~$j$ having~$s(j)=1$, then with
        arguments similar to those of the base case, we~get
	$b_{\lst{\rho},d_\rho,s} \preceq_s f^w$.
	
      \item otherwise, the goals enabled by~$s$ split following an
        extended partition
        $\tau=(s_\kappa,q_\kappa,d_\kappa)_{\kappa\leq \lambda}$.  We
        let~$\eta$ be the history from the last state of~$\rho$ to the
        point where the goals split.
	
        
	
	The map $\vartheta_{\lst{\rho},d_\rho,s}$ witnesses that $\calG,\lst{\rho}
	\models^{T} \Gamma_{d,s,b_{\lst{\rho},d_\rho,s}}$;
	therefore $\eta$ may only reach
	a partition~$\tau$ such that
	\begin{equation}
	  b_{\lst{\rho},d_\rho,s} \preceq_s c_{s,\tau} 
	  \label{Label::3:Mars:2017::numero1::Preuve::Prop23}
	\end{equation}
	This partition~$\tau$ is such that 
	for any $1\leq \kappa \leq \lambda$, 
        it~holds $(s_\kappa,\rho\cdot \eta \cdot q_\kappa)\in R^w$;
        using the induction hypothesis, we~get
	\begin{equation}
	  s_\kappa\inter b_{q_\kappa,d_\kappa,s_\kappa} \preceq_{s_\kappa} f^w
	  \label{Label::3:Mars:2017::numero2::Preuve::Prop23}
	\end{equation}
	
	Then, using Lemma~\ref{LemmaOrderGetCoarserwithS}
        (see~Appendix~\ref{app-lemma18}) repeatedly on the
        $(s_\kappa)_{1\leq \kappa\leq\lambda}$, and
        Equation~\eqref{Label::3:Mars:2017::numero2::Preuve::Prop23},
        we obtain
	\begin{alignat*}1
	   s_1\inter b_{q_1,d_1,s_1} \preceq_{s_1} f^w
	  &\impl  (s_1  \inter b_{q_1,d_1,s_1}) \union (s_2 \inter b_{q_2,d_2,s_2}) \preceq_{s_1 \union s_2} f^w
	  \\
	  & \impl \dots 
	  \\
	  &\impl  (s_1  \inter b_{q_1,d_1,s_1}) \union \dots \union (s_\lambda \inter b_{q_\lambda,d_\lambda,s_\lambda}) \preceq_{s_1 \union  \dots  \union s_\lambda} f^w
	  \\
	  &\impl c_{s,\tau} \preceq_s f^w.
	\end{alignat*}
	Combined with~\eqref{Label::3:Mars:2017::numero1::Preuve::Prop23}, 
	we get $b_{\lst{\rho},d_\rho,s}\preceq_s  c_{s,\tau} \preceq_s f^w$.  \popQED
      \end{itemize}
      \let\qed\relax
    \end{proof}
    
    \begin{lemma}
      \label{lemma-zpazjfazekbfk}
      $\calG,q_0\models_\chi \Gamma_{F^n}$ if, and only~if,
      $b_{q_0,d_0 , \bfone} \in F^n$.
    \end{lemma}

    \begin{proof}
      Assume that $b_{q_0,d_0 , \bfone} \in
      \overline{F^n}$.
      Then 
      $\calG,q_0\not\models^{T} (Q_i x_i)_{1\leq i\leq l}.\ \Gamma_{F^n}$.
      Applying
      Lemma~\ref{LemmaBidexistsmapsvarrhosProofThmModelCheck}, 
      the map $\overline{\vartheta_\bfone}$ (and therefore
      $\overline{\theta}$, which act as $\overline{\vartheta_\bfone}$
      before goals branch along different paths) witnesses $\calG,q_0
      \not\models^{T} (Q_i x_i)_{1\leq i\leq l}.\ \Gamma_{F^n}$.
      This~implies that $\calG,q_0
      \not\models_{\chi} \Gamma_{F^n}$, which contradicts 
      the hypothesis.

      Conversely, if $b_{q_0,d_0 , \bfone} \in F^n$, then
      $\calG,q_0\models^{T} (Q_i x_i)_{1\leq i\leq l}.\ \Gamma_{F^n}$,
      which is witnessed by map~$\vartheta_1$. Thus
      $\calG,q_0\models_\chi \Gamma_{F^n}$.
    \end{proof}
    
    We are now ready to prove the first part of
    Lemma~\ref{lemmaintermediaireavantElemeoptim}: consider a function
    $w\colon \calV^\forall \to (\Hist_\calG \to \Act)$.
    By~Lemma~\ref{lemma-FinitionThmModelCheckOptimal} applied to~$w$,
    $s=\bfone$, and~$\rho=q_0$, we get that $ b_{q_0,d_0,\bfone}
    \preceq_\bfone f^w$.
    Now, by Lemma~\ref{lemma-zpazjfazekbfk}, $b_{q_0,d_0 , \bfone} \in
    F^n$, therefore the element $f^w$, being greater than $
    b_{q_0,d_0,\bfone}$ for~$\preceq_\bfone$, must also be in~$F^n$,
    which means that $\calG,q_0 \models_{\theta(w)}
    F^n(\beta_j\phi_j)_{1\leq j\leq n}$.
    
    \medskip
    
    The second implication of the lemma is proven using similar arguments.
  \end{proof}
  
  Lemma~\ref{lemmaintermediaireavantElemeoptim} allows us to conclude
  that at least one of $\phi$ and~$\neg \phi$ must hold on~\calG
  for~$\models^{T}$.
  Lemma~\ref{PropNegationsAndCo::DeuxTroisAvril::ccr} implies that at
  most one can hold. Combining both we get that exactly one holds.
\end{proof}

\subsection{Proof of Proposition~\ref{ch:EG::PropMaximaliteSLEGvisavisSLBG}}
\MaximalitySLEGThmRestatableThingymarkOEOS*

\begin{proof}
  We consider the game~$\calG$ depicted on Figure~\ref{fig-g} with two agents \BoxFill[\FigColB]{} and \CircleFill[\FigColA]{}.
  Let $F^n$ be a non-\semist set over $\{0,1\}^n$. Then there must exist $f_1, f_2 \in
  F^n$, and $s \in \ZO^n$, such that
  ${(f_1 \inter s) \union (f_2 \inter \bar s) \notin F^n}$
  and $(f_2 \inter s) \union (f_1 \inter \bar s) \notin F^n$.
  We then let
  \[
  \phi =
  \forall y^{\BoxUnderFill[\FigColB]{}}_t.\ \forall y^{\BoxUnderFill[\FigColB]{}}_{u}.\ \forall x^{\CircleUnderFill[\FigColA]{}}_t.\ \exists x^{\CircleUnderFill[\FigColA]{}}_{u}
  .\ F^n(\beta_1 \varphi_1,\dots,\beta_n \varphi_n)
  \]
  where 
  \[
  \beta_i = 
  \begin{cases}
    \assign{\BoxFill[\FigColB]{}\mapsto y^{\BoxUnderFill[\FigColB]{}}_t;\ \CircleFill[\FigColA]{}\mapsto x^{\CircleUnderFill[\FigColA]{}}_t}  &\text{if $s(i)=1$}\\
    \assign{\BoxFill[\FigColB]{}\mapsto y^{\BoxUnderFill[\FigColB]{}}_{u};\ \CircleFill[\FigColA]{}\mapsto x^{\CircleUnderFill[\FigColA]{}}_{u}} &\text{if $s(i)=0$}
  \end{cases}
  \]
  and
  \[
  \varphi_i = \left\{\begin{array}{ll}
    \F p_1 \vee \F p_2 & \text{if}\ f_1(i)=f_2(i)=1 \\
    \F p_1 & \text{if}\ f_1(i)=1\ \text{and}\ f_2(i)=0 \\
    \F p_2 & \text{if}\ f_1(i)=0\ \text{and}\ f_2(i)=1 \\
    \texttt{false} & \text{if}\ f_1(i)=f_2(i)=0
  \end{array}\right.
  \]
  
    \begin{figure}[bt]
    \centering
    \begin{tikzpicture}[yscale=.9]
      \draw (0,0) node [draw=\FigColB, fill=\FigColB!40!white, rectangle] (q0) {$q_0$};
      \draw (-2,-1) node [draw=\FigColA, fill=\FigColA!40!white, circle] (q1) {$q_t$};
      \draw (2,-1) node [draw=\FigColA, fill=\FigColA!40!white, circle] (q2) {$q_{u}$};
      \draw (-3,-2) node [draw=black!40!white, circle] (q1f1) {$q_{t1}$} node [left=11pt] {$p_1$};
      \draw (-1,-2) node [draw=black!40!white, circle] (q1f2) {$q_{t2}$} node [left=11pt] {$p_2$};
      \draw (1,-2) node [draw=black!40!white, circle] (q2f1) {$q_{u1}$} node [left=11pt] {$p_1$};
      \draw (3,-2) node [draw=black!40!white, circle] (q2f2) {$q_{u2}$}node [left=11pt] {$p_2$};
      
      \foreach \n in {q1f1,q1f2,q2f1,q2f2}
      {\draw (\n) edge[out=-60,in=-120,looseness=6,-latex'] (\n);}
      \draw [-latex'] (q0) -- (q1);
      \draw [-latex'] (q0) -- (q2);
      \draw [-latex'] (q1) -- (q1f1);
      \draw [-latex'] (q1) -- (q1f2);
      \draw [-latex'] (q2) -- (q2f1);
      \draw [-latex'] (q2) -- (q2f2);
    \end{tikzpicture}
    \caption{The two-agents turn-based game~$\calG$}
    \label{fig-g}
  \end{figure}
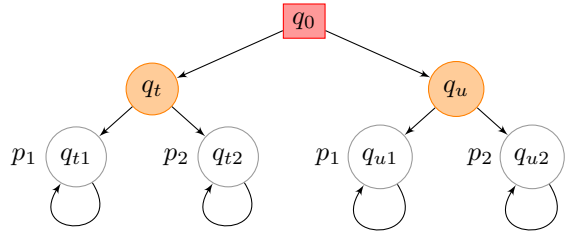
  
  It~is not hard to check that the following holds:
  \begin{lemma}
    \label{lemma:intuition}
    Let $\rho$ be a maximal run of $\mathcal{G}$ from~$q_0$. 
    Let~$k\in\{1,2\}$ be such that $\rho$ visits a state labelled with~$p_k$.
    Then
    for~any~$1\leq i\leq n$, we~have $\rho \models \varphi_i$ if, and only~if,
    $f_k(i)=1$. 
  \end{lemma}

  The following two lemmas conclude the proof:
  \begin{lemma}
    $\calG, q_0 \not\models^{T} \phi$
    \label{lemma-ThemMaximaliteSLEGPropTwyou}
  \end{lemma}
  \begin{proof}
    Towards a contradiction, assume that $\calG, q_0 \models^{T} \phi$.    
    We let $\sigma_t$ (resp.~$\sigma_{u}$) be the strategy that maps
    history~$q_0$ to $q_t$  (resp.~$q_{u}$). We fix strategy $\tau_t$ such
    that $\tau_t(q_0\cdot q_t) = 
    q_{t1}$. There is a strategy
    $\tau_{u}$ (with local and timeline dependences) such that
    \[
    \calG, q_0 \models_\chi F^n(\beta_1 \varphi_1,\dots,\beta_n \varphi_n)
    \]
    where $\chi$ maps~$y^{\BoxUnderFill[\FigColB]{}}_t$ to~$\sigma_t$, $y^{\CircleUnderFill[\FigColA]{}}$~to~$\sigma_u$,
    $x^{\CircleUnderFill[\FigColA]{}}_t$~to~$\tau_t$ and $x^{\CircleUnderFill[\FigColA]{}}_u$~to~$\tau_u$.
    
    Since $x^{\CircleUnderFill[\FigColA]{}}_{u}$ is only jointly applied with~$y^{\BoxUnderFill[\FigColB]{}}_{u}$, the only important
    information about $\tau_{u}$ is its value on history $q_0
    \sigma_{u}(q_0) = q_0 q_{u}$. This value is then independent on the
    value of $\tau_t(q_0 q_t) = \tau_t(q_0 \sigma_t(q_0))$. In~particular,
    writing~$\chi'$ for the context obtained from~$\chi$ by replacing $\chi(y^{\BoxUnderFill[\FigColB]{}}_t)=\tau_t$
    with~$\tau'_t$, where $\tau'_t(q_0 q_t) = q_{t2}$, we~also have
    \[
    \calG, q_0 \models_{\chi'} F^n(\beta_1 \varphi_1,\dots,\beta_n \varphi_n)
    \]
    
    Let $v$ and $v'$ be the vectors in~$\ZO^n$
    representing the values of the goals $(\beta_1 \varphi_1,\dots,\beta_n
    \varphi_n)$ under~$\chi$ and~$\chi'$.
    Then $v$ and $v'$ are in~$F^n$. However:
    \begin{itemize}
      \item If $\tau_{u}(q_0 q_{u}) = q_{u1}$, then
      $v' = (f_1 \inter \bar s) \union (f_2 \inter s)$.
      \item If $\tau_{t}(q_0 q_{t}) = q_{t2}$, then
      $v = (f_1 \inter s) \union (f_2 \inter \bar s)$.
    \end{itemize}  
    In both cases, by hypothesis, this does not belong to~$F^n$, which is
    a contradiction. 
  \end{proof}

  \begin{lemma}
    $\calG ,q_0 \not\models^{T} \neg \formule$
    \label{lemma-ThemMaximaliteSLEGPropTrwe}
  \end{lemma}
  
  \begin{proof}
    Similarly, assume $\calG ,q_0 \models^{T} \neg \formule$.
    Fix any three strategies $\sigma_t$, $\sigma_{u}$ and $\tau_t$ 
    respectively intended for the existentially quantified variables 
    $y^{\BoxUnderFill[\FigColB]{}}_t$, $y^{\BoxUnderFill[\FigColB]{}}_{u}$ and $x^{\CircleUnderFill[\FigColA]{}}_t$.
    Due to the nature of $\models^{T}$, these three strategies are independent from the strategy $\tau_u$ of $x^{\CircleUnderFill[\FigColA]{}}_{u}$.
    Consider then the following strategy $\tau_u$:
    \[
     \tau_u(q_0.\sigma_u(q_0))= \tau_t(q_0.\sigma_t(q_0))
     \]
     Let $\chi$ be the resulting context and $v$ the vector representing the values of the goals $(\beta_1 \varphi_1,\dots,\beta_n
    \varphi_n)$ under~$\chi$. Either $v=f_1$ or $v=f_2$; in both case $v\in F^n$, which is a contradiction.
  \end{proof}
\let\qed\relax
\end{proof}





%
%
%

\end{document}